\documentclass{article}
\usepackage{times,natbib}
\setcitestyle{authoryear,round,citesep={;},aysep={,},yysep={;}}
\usepackage[textwidth=5.5in,textheight=9in]{geometry}
\usepackage{amsmath,amsfonts,bm}

\def\eqref#1{equation~\ref{#1}}

\def\1{\bm{1}}

\DeclareMathAlphabet{\mathsfit}{\encodingdefault}{\sfdefault}{m}{sl}
\SetMathAlphabet{\mathsfit}{bold}{\encodingdefault}{\sfdefault}{bx}{n}

\newcommand{\E}{\mathbb{E}}

\newcommand{\Var}{\mathrm{Var}}

\usepackage{float}
\usepackage{hyperref}
\usepackage{url}
\usepackage{longtable}
\hypersetup{hidelinks}

\usepackage{colortbl}
\usepackage{amsmath,amssymb,amsthm}

\usepackage{booktabs}
\usepackage{csvsimple-l3}
\ExplSyntaxOn
\tl_new:N \l__tablecount_digits_tl
\int_new:N \l__tablecount_left_int
\NewDocumentCommand \tablecount { m }
  {
    \tl_set:Ne \l__tablecount_digits_tl {#1}
    \int_set:Nn \l__tablecount_left_int { \tl_count:N \l__tablecount_digits_tl }
    \tl_map_inline:Nn \l__tablecount_digits_tl
      {
        ##1
        \int_decr:N \l__tablecount_left_int
        \bool_lazy_and:nnT
          { \int_compare_p:nNn { \l__tablecount_left_int } > { 0 } }
          { \int_compare_p:nNn { \int_mod:nn { \l__tablecount_left_int } { 3 } } = { 0 } }
          { {,} }
      }
  }
\ExplSyntaxOff
\usepackage{array}
\usepackage{enumitem}

\usepackage{graphicx}
\usepackage{adjustbox}
\usepackage{listings}
\usepackage{xspace}
\usepackage[normalem]{ulem}
\DeclareRobustCommand{\revdel}[1]{{\setlength{\fboxsep}{0.6pt}%
  \colorbox{red!20}{\sout{#1}}}}

\DeclareRobustCommand{\bmrevdel}[1]{\textcolor{violet}{\sout{#1}}}

\usepackage{tikz}
\usetikzlibrary{arrows.meta,positioning,fit,calc}
\usepackage{pifont}  
\usepackage{mdframed}
\newmdenv[backgroundcolor=orange!18,hidealllines=true,nobreak=true,
  innerleftmargin=3pt,innerrightmargin=3pt,
  innertopmargin=3pt,innerbottommargin=3pt,
  skipabove=3pt,skipbelow=3pt]{tonereview}
\newmdenv[backgroundcolor=blue!8,hidealllines=true,
  innerleftmargin=3pt,innerrightmargin=3pt,
  innertopmargin=3pt,innerbottommargin=3pt,
  skipabove=3pt,skipbelow=3pt]{bytemarkrevision}
\newmdenv[backgroundcolor=red!15,hidealllines=true,
  innerleftmargin=3pt,innerrightmargin=3pt,
  innertopmargin=3pt,innerbottommargin=3pt,
  skipabove=3pt,skipbelow=3pt]{searchreview}

\newtheorem{proposition}{Proposition}
\newtheorem{theorem}{Theorem}

\newcommand{\ByteMark}{ByteMark\xspace}

\newcommand{\PRF}{F}   

\providecommand{\E}{\mathbb{E}}
\providecommand{\Var}{\mathrm{Var}}

\providecommand{\1}{\mathbf{1}}

\lstdefinestyle{pycode}{
  language=Python,
  basicstyle=\ttfamily\footnotesize,
  keywordstyle=\color{blue!70!black}\bfseries,
  commentstyle=\color{gray},
  stringstyle=\color{red!60!black},
  numbers=left,
  numberstyle=\tiny\color{gray},
  breaklines=true,
  showstringspaces=false,
  frame=single,
  framerule=0.3pt,
  xleftmargin=1.5em,
  columns=fullflexible
}
\usepackage{algorithm}
\usepackage{algpseudocode}

\makeatletter
\newcommand{\historicaltable}[1]{%
  \begingroup
  \let\historicalmakecaption\@makecaption
  \long\def\@makecaption##1##2{%
    \historicalmakecaption{\textcolor{red}{##1}}{%
      \textcolor{red}{\textbf{Historical data (layout only).} ##2}}}%
  \input{#1}%
  \endgroup
}
\makeatother


\newcommand{\experimenttable}[1]{%
  \ifcsname papertable@#1\endcsname
    \csname papertable@#1\endcsname
  \else
    \PackageError{ByteMark}{Table selection missing: #1}{Build with latexmk to refresh table selections.}%
  \fi
}
\newcommand{\maintextdetailtable}[1]{%
  \begingroup
  \let\detailSavedLabel\label
  \renewcommand{\label}[1]{\detailSavedLabel{main-detail:##1}}%
  \let\detailSavedTable\table
  \let\detailSavedEndTable\endtable
  \renewenvironment{table}[1][]{\detailSavedTable[H]}{\detailSavedEndTable}%
  \renewenvironment{table*}[1][]{\begin{table}[H]}{\end{table}}%
  \setlength{\tabcolsep}{4pt}%
  \csname papertablemain@#1\endcsname
  \endgroup
}
\newcommand{\maintexttable}[1]{%
  \ifcsname papertablemain@#1\endcsname
    \csname papertablemain@#1\endcsname
  \else
    \PackageError{ByteMark}{Main-text table selection missing: #1}{Build with latexmk.}%
  \fi
}

\DeclareRobustCommand{\claudedel}[1]{\textcolor{red}{\sout{#1}}}
\DeclareRobustCommand{\claudeadd}[1]{\textcolor{red}{#1}}
\title{ByteMark: A Unified Byte-Level Watermark\\
for Detection Across Language Models}

\author{Xiaopu Wang\\
\texttt{xmw5221@psu.edu}
\and
Runze Li\\
\texttt{rli4@psu.edu}}

\date{}

\begin{document}

\maketitle

\begin{abstract}
As watermarking of LLM-generated text enters deployment, users increasingly
need to check text whose source is unknown. Existing detectors read the
watermark through a specific tokenizer or embedding pipeline, and even if all
providers adopted the same watermarking method, differences in these
representations would leave their watermarks incompatible, and the text alone
does not reveal which representation produced it. We propose \ByteMark, a
{\fontencoding{T1}\selectfont\textit{unified watermarking scheme}} that
defines its signal on the UTF-8 bytes of published text. Models with different
tokenizers embed this signal under one public specification and key, and a
single detector verifies the text from the text and key alone. When the
tokenizer or embedding model of a baseline detector is replaced, its detection
degrades sharply, whereas \ByteMark, which uses neither, retains its detection
rate; \ByteMark also remains detectable when participating models continue or
rewrite the text.
\end{abstract}

\begingroup
\raggedbottom
\setlength{\textfloatsep}{12pt plus 2pt minus 2pt}
\setlength{\intextsep}{10pt plus 2pt minus 2pt}
\setlength{\floatsep}{10pt plus 2pt minus 2pt}
\setlength{\abovecaptionskip}{6pt}
\let\mainSavedTable\table
\let\mainSavedEndTable\endtable
\renewenvironment{table}[1][]{\mainSavedTable[!htbp]}{\mainSavedEndTable}

\section{Introduction}
\label{sec:introduction}

Large language models (LLMs) generate text for question answering, writing
assistance, and other instruction-driven
tasks~\citep{brown2020language,ouyang2022training,dathathri2024synthid}. When
generated text is presented as human-written work, it can undermine authorship
assessment and mislead readers about the origin of content
\citep{kirchenbauer2023watermark,kuditipudi2023robust,dathathri2024synthid}.
LLMs can also hallucinate plausible but factually incorrect content,
potentially misleading readers \citep{lin2022truthfulqa}. These risks also
extend to training-data curation: repeatedly training on generated text can
lead to \emph{model collapse}, reducing the diversity and performance of
subsequent models~\citep{shumailov2024curse,kazdan2025collapse}. Identifying
whether text was generated by an LLM can help assess claims of human
authorship and track the composition of training data. Publication venues
have already taken steps to identify LLM-generated text: ICML has used prompt injection to detect LLM
use in peer reviews, and the NeurIPS 2026 Position Paper Track has screened
submissions for AI-generated content
\citep{icml2026reviewwatermark,neurips2026aipapers}.

Such identification nonetheless remains difficult. An instruction hidden in a
paper, such as a prompt injection, can be found and removed by the
reviewer~\citep{liu2026incontext}. Zero-shot detectors such as DetectGPT rely
on linguistic differences and can lose accuracy with a surrogate scoring model
or under a different generator and decoding
strategy~\citep{mitchell2023detectgpt,dugan2024raid}, and human readers do not
reliably distinguish LLM-generated responses from human
ones~\citep{jones2024turing}.

A complementary approach is watermarking. By deliberately embedding a
recognizable signal into generated text, watermarking gives verifiers evidence
that can be tested, without relying solely on incidental
differences between human and generated writing
\citep{kirchenbauer2023watermark,kuditipudi2023robust}. Regulation has begun to mandate such marking: the European Union AI
Act, applicable from August 2026, requires that generated outputs be ``marked
in a machine-readable format and detectable as artificially
generated''~\citep[Art.~50]{eu2024aiact}. We therefore expect providers of generative models to \emph{participate} in some form of watermarking.

Existing watermarks embed this signal in tokens or in sentences. Token-level
methods alter token selection through logit bias or keyed sampling, and their
detectors aggregate keyed token statistics
\citep{kirchenbauer2023watermark,kuditipudi2023robust,dathathri2024synthid}.
These detectors do not need the generator's weights and are often described
as model-agnostic
\citep{kirchenbauer2023watermark,chen2025mcmark,wu2024dipmark}, yet their
signal is defined on token IDs, which differ across tokenizers for the same
text. Sentence-level methods map candidate sentences to a semantic embedding
and resample until a candidate satisfies the watermark condition; the
discarded candidates reduce computational efficiency and increase latency
\citep{hou2024semstamp,hou2024ksemstamp,li2025simmark,li2026alimark}. Their
detectors require the embedding pipeline used to define the watermark, and some
also require stored artifacts derived from it, such as cluster centroids
\citep{hou2024ksemstamp,li2025simmark}. In both families, detection therefore
depends on a representation that the text does not reveal and that, for
closed-source models, may not be publicly available
(Section~\ref{sec:problem}).

This dependence does not match practical verification. An area chair checking
a peer review, an instructor assessing student work, or a curator filtering
training data typically receives the text without knowing which model, if any,
produced it. Whoever runs the test, the reader or a service acting for the
reader, holds the text and the watermark key but not the generating model or
its components. We therefore distinguish two detection tasks:

\begin{description}[style=nextline,leftmargin=2em,nosep,font=\normalfont\itshape]
 \item[Presumed-source verification.] Verifying a watermark associated with a presumed source, given its
 tokenizer or embedding pipeline.
 \item[Unknown-source detection.] Detecting a watermark from any participating
 model without knowing how the text was produced.
\end{description}
In unknown-source detection, a negative result under one assumed
representation cannot establish that a passage is unwatermarked unless the
verifier has searched among candidates known to include the generating representation.

These representations multiply as watermarking enters deployment.
Token-level watermarking is already deployed in Gemini and Claude
\citep{dathathri2024synthid,anthropic2026watermark}. \emph{Even if all
providers adopted the same watermarking method, the method alone would not
make their watermarks compatible}, because each provider applies it through
its own tokenizer or embedding pipeline. A single document may also combine
several models: multi-agent systems assign the writing and revision of shared
text to agents that need not run the same model~\citep{wu2024autogen,wang2025moa}.
Training corpora also collect text written over many years, during which
tokenizers and embedding pipelines change.

We therefore seek a common watermark that participating providers embed under
a shared key, with three properties.

\begin{itemize}[leftmargin=2em,nosep]
 \item \textbf{Model independence.} Detection should depend only on the
 text, the watermark key, and a fixed public specification, so that one
 detection rule applies across participating models and gives the same result for a given text
 as models, tokenizers, and embedding pipelines change.
 \item \textbf{Continuity across models.} The same detector should remain
 effective when participating models continue or rewrite the text.
 \item \textbf{Low generation overhead.} Because deployed language models
 serve a large volume of requests, embedding should add limited computation
 and latency during generation.
\end{itemize}

To meet these requirements, we propose \ByteMark, which defines its signal on
the UTF-8 byte stream of the published text. During generation, each
participating model biases its token selection with a shared byte-level
scoring rule; at detection, a verifier applies the
same rule to the text's bytes using the shared key and public specification.
Because detection requires neither a tokenizer nor a semantic pipeline, one
detector serves every participating model, and \ByteMark{} functions as a
{\fontencoding{T1}\selectfont\textit{unified watermarking scheme}} for
heterogeneous generators.

\paragraph{Contributions.}
\begin{enumerate}[leftmargin=2em,nosep]
 \item We distinguish presumed-source verification from unknown-source detection and show empirically that replacing the tokenizer or embedding model removes most of the signal of existing token-level and sentence-level detectors, limiting their use for unknown-source detection.
 \item We propose \ByteMark, which defines the watermark on the UTF-8 byte stream
 so that heterogeneous generators share one signal definition and one
 detection rule. By construction, the detection score is determined solely by the text, key, and public specification, without reference to the generating tokenizer. Under an ideal pseudorandom function, the detection
 statistic has an exact finite-sample null distribution, which yields
 thresholds at any target false-positive rate without calibration data.
 \item We demonstrate model independence and continuity across models by applying a
 single \ByteMark{} detector to outputs of heterogeneous generators and to text
 continued or rewritten by other participating models, without selecting a
 representation at model boundaries. When a participating model rewrites the
 text and re-embeds the watermark, the rewritten output remains detectable.
 We also quantify generation overhead through the additional candidate-token work and
 latency of sentence-level methods alongside \ByteMark's own byte-scoring
 overhead.
\end{enumerate}


\section{Related Work}
\label{sec:related}
\label{sec:detection-requirements}
\label{sec:problem}

We review token-level and sentence-level watermarking methods in terms of their
generation procedures and detection requirements. A watermark detector can
operate without querying the generating model yet still depend on its
representation pipeline. We examine how embedding coordinates, implementation
configurations, and token IDs enter the detection statistic and how changing
the pipeline affects detection.

\paragraph{Sentence-level watermarking.}
Sentence-level watermarking methods typically rely on a semantic embedding
pipeline to select candidate sentences and detect the embedded signal.
Their detection statistics are defined in the geometry of that embedding space:
hyperplane regions \citep{hou2024semstamp}, $k$-means or fuzzy clusters
\citep{hou2024ksemstamp,diaa2026kswordstamp,zhang2025cohemark}, similarity
intervals between adjacent sentences \citep{li2025simmark}, and projections onto
keyed pivot vectors or directions \citep{wang2026pmark,huo2026samark,li2026alimark}.
Two embedding models can capture similar relational structure without preserving
these hyperplane regions, clusters, similarity intervals, or pivot
projections. Our embedding-model
replacement experiments show score degradation
(Table~\ref{tab:e5}), and the detectors are undefined for encoders with a different output dimension (Appendix~\ref{sec:encoder-replacement}).
When such a watermark is deployed, this dependence creates a long-term
compatibility trade-off:

\begin{itemize}[leftmargin=2em,nosep]
 \item \textbf{Retain the original pipeline.} A fixed embedding pipeline
 preserves compatibility with previously embedded watermarks
 but cannot benefit from newer embedding models with improved text representations
 \citep{wang2024improvingembeddings,lee2025nvembed,zhang2025qwen3embedding}.
 Retaining the original pipeline also requires maintaining an aging
 deployment environment, whose software dependencies may eventually become deprecated
 or unsupported.
 \item \textbf{Update the pipeline.} A new pipeline may fail to detect
 watermarks embedded with the original one (Section~\ref{sec:exp-replacement}), so updating risks losing backward compatibility.
\end{itemize}

Sentence-level resampling and candidate selection can generate
multiple continuations for each retained sentence, increasing the number of
tokens generated per final output token
\citep{hou2024semstamp,hou2024ksemstamp,li2025simmark}. This extra generation
can increase latency and reduce final-output throughput, depending on how
candidate generation is batched or parallelized. Such overhead matters at the
scale of current use: agentic coding tasks consume on the order of a thousand
times more tokens than code chat or reasoning tasks~\citep{bai2026agents}, and
token consumption is a major computational and economic cost
for deployed systems~\citep{chen2026token}. Section~\ref{sec:exp-efficiency}
evaluates throughput and candidate-token work separately.

\paragraph{Token-level watermarking.}
Token-level watermarking methods embed the signal at each decoding step by
modifying the token probabilities or the sampling rule. KGW and its
descendants bias a keyed green region in the model vocabulary
\citep{kirchenbauer2023watermark,zhao2024unigram,kirchenbauer2024reliability}.
Distortion-free reweighting and deterministic sampling alter the sampling rule
but retain randomness indexed by the token context or position
\citep{kuditipudi2023robust,wu2024dipmark,hu2024unbiased,feng2025bimark};
SynthID adds a multi-layer tournament signal \citep{dathathri2024synthid}.
Semantic-conditioned variants map an embedding model's output to token-level
biases through learned or projected transforms, and therefore retain a
vocabulary-specific mapping
\citep{liu2024semantic,he2024xsir,ren2024semamark,schafer2026dew}.

The same text can receive different token IDs and segmentations under
different tokenizers (Table~\ref{tab:tokenizer-example}). For token-based
watermarks, these differences change the inputs to the detection rule and can
therefore change the detector's keyed assignments and score.

\begin{table}[H]
\centering
\small
\setlength{\tabcolsep}{4pt}
\caption{Tokenization of three words under six fixed tokenizers. Each input has a leading space and is encoded without special tokens. Each cell shows token pieces (top) and IDs (bottom); centered dots separate tokens.}
\label{tab:tokenizer-example}
\maxsizebox{\textwidth}{!}{%
\begin{tabular}{lcccccc}
\toprule
Word & Pythia & Qwen2.5 & OpenLLaMA & GPT-Neo & OPT & TinyLlama \\
\midrule
\texttt{model} & \begin{tabular}[c]{@{}c@{}}\texttt{model}\\{\scriptsize 1566}\end{tabular} & \begin{tabular}[c]{@{}c@{}}\texttt{model}\\{\scriptsize 1614}\end{tabular} & \begin{tabular}[c]{@{}c@{}}\texttt{model}\\{\scriptsize 2042}\end{tabular} & \begin{tabular}[c]{@{}c@{}}\texttt{model}\\{\scriptsize 2746}\end{tabular} & \begin{tabular}[c]{@{}c@{}}\texttt{model}\\{\scriptsize 1421}\end{tabular} & \begin{tabular}[c]{@{}c@{}}\texttt{model}\\{\scriptsize 1904}\end{tabular} \\
\texttt{watermark} & \begin{tabular}[c]{@{}c@{}}\texttt{waterm$\cdot$ark}\\{\scriptsize 37385$\cdot$782}\end{tabular} & \begin{tabular}[c]{@{}c@{}}\texttt{watermark}\\{\scriptsize 88006}\end{tabular} & \begin{tabular}[c]{@{}c@{}}\texttt{water$\cdot$mark}\\{\scriptsize 1536$\cdot$3630}\end{tabular} & \begin{tabular}[c]{@{}c@{}}\texttt{water$\cdot$mark}\\{\scriptsize 1660$\cdot$4102}\end{tabular} & \begin{tabular}[c]{@{}c@{}}\texttt{water$\cdot$mark}\\{\scriptsize 514$\cdot$6920}\end{tabular} & \begin{tabular}[c]{@{}c@{}}\texttt{wat$\cdot$erm$\cdot$ark}\\{\scriptsize 16699$\cdot$837$\cdot$935}\end{tabular} \\
\texttt{paraphrase} & \begin{tabular}[c]{@{}c@{}}\texttt{par$\cdot$aphr$\cdot$ase}\\{\scriptsize 1061$\cdot$24596$\cdot$511}\end{tabular} & \begin{tabular}[c]{@{}c@{}}\texttt{paraph$\cdot$rase}\\{\scriptsize 62230$\cdot$10632}\end{tabular} & \begin{tabular}[c]{@{}c@{}}\texttt{par$\cdot$aphr$\cdot$ase}\\{\scriptsize 1302$\cdot$29486$\cdot$542}\end{tabular} & \begin{tabular}[c]{@{}c@{}}\texttt{paraph$\cdot$rase}\\{\scriptsize 38624$\cdot$22789}\end{tabular} & \begin{tabular}[c]{@{}c@{}}\texttt{paraph$\cdot$rase}\\{\scriptsize 40127$\cdot$34338}\end{tabular} & \begin{tabular}[c]{@{}c@{}}\texttt{par$\cdot$ap$\cdot$hr$\cdot$ase}\\{\scriptsize 610$\cdot$481$\cdot$1092$\cdot$559}\end{tabular} \\
\bottomrule
\end{tabular}}
\end{table}

Section~\ref{sec:exp-replacement} measures this effect systematically (Table~\ref{tab:compact-mismatch-token}), and Appendix~\ref{sec:related-appendix} discusses watermarking through formatting,
synonym substitution, and prompting.


\section{Method}
\label{sec:method}

Section~\ref{sec:detection-requirements} traced the difficulty to the carrier:
embedding coordinates, token IDs, and their implementation configurations
exist only inside components the verifier may not hold. We therefore look for
a watermark carrier inside the text itself. We introduce that construction
through a toy example, LetterMark, which uses the letters of English text. It serves only as an illustration, yet it already detects from the text alone; Section~\ref{sec:bytemark} applies the
same pattern to a carrier that covers arbitrary text.

\subsection{Toy example: LetterMark}
\label{sec:toy-example}

LetterMark works as follows (Figure~\ref{fig:lettermark-example}). At each
generation step, a keyed hash of the last $h$ letters of the completion
selects a green list of letters. Since a token is generated as a whole but may
contain several letters, each candidate token is scored letter by letter with
a rolling window, and its score $R$ is
its fraction of green letters. The fraction $\gamma$ of candidates with the
highest $R$ receives a logit increase of $\delta>0$, and the model samples as
usual with its own tokenizer.

The verifier holds only the published text and the key. It slides the same
window over the letters, recomputes each green list, and counts green letters;
an unusually high rate indicates the watermark.

\begin{figure}[!htb]
  \centering
  \resizebox{.8\linewidth}{!}{%
  \begin{tikzpicture}[
      font=\small,
      >={Stealth[length=4pt]},
      node distance=3mm and 4mm,
      cell/.style={draw, minimum width=5.6mm, minimum height=6mm,
                   inner sep=0pt, font=\ttfamily\small},
      inw/.style={cell, fill=blue!7, draw=blue!55!black},
      outw/.style={cell, fill=gray!8, draw=gray!60, text=gray!70},
      box/.style={draw, rounded corners=2pt, minimum height=7mm,
                  inner sep=4pt, align=center},
      hash/.style={box, fill=gray!12},
      glist/.style={box, fill=gray!5},
      hit/.style={box, fill=green!25, draw=green!45!black, font=\ttfamily\small},
      miss/.style={box, fill=red!12, draw=red!55!black, font=\ttfamily\small},
      win/.style={draw=blue!55!black, thick, rounded corners=2pt, inner sep=1.5pt},
      lab/.style={font=\small\bfseries, anchor=east},
      note/.style={font=\footnotesize, align=left},
    ]
    \def\dy{1.35}
    \def\cw{0.6}
    \newcommand{\step}[9]{%
      \node[lab] at (-0.4,#1) {step #2};
      \foreach \i/\c/\s in {#3}
        \node[\s] (s#2\i) at (\cw*\i,#1) {\c};
      \node[win, fit=(s#2#4)(s#2#5)] (w#2) {};
      \node[hash] (h#2) at (5.3,#1) {$\PRF_k(\cdot)$};
      \node[glist, right=of h#2] (g#2) {$\{#6\}$};
      \node[#8, right=of g#2] (c#2) {#7};
      \node[note, right=2mm of c#2] {#9};
      \draw[->] (w#2.east) -- (h#2);
      \draw[->] (h#2) -- (g#2);
      \draw[->] (g#2) -- (c#2);
    }
    \node[note, anchor=center, align=center] at (4.7,\dy-0.2)
      {\textbf{Example:} scoring the candidate token \texttt{tea}\\
       window of last $h{=}4$ letters $\;\rightarrow\;$ green list
       $\;\rightarrow\;$ is the next letter green?};
    \step{0}{1}{0/c$_1$/inw,1/c$_2$/inw,2/c$_3$/inw,3/c$_4$/inw}{0}{3}
         {e,a,t}{t}{hit}{\texttt{t} green}
    \step{-\dy}{2}{0/c$_1$/outw,1/c$_2$/inw,2/c$_3$/inw,3/c$_4$/inw,4/t/inw}{1}{4}
         {r,o,n}{e}{miss}{\texttt{e} not green}
    \step{-2*\dy}{3}{0/c$_1$/outw,1/c$_2$/outw,2/c$_3$/inw,3/c$_4$/inw,4/t/inw,5/e/inw}{2}{5}
         {a,s,i}{a}{hit}{\texttt{a} green}
    \draw[->, blue!55!black, thick] (w1.south) ++(0.3,0) -- ++(0,-0.45)
      node[midway, right, note, text=blue!55!black] {append \texttt{t}, drop c$_1$};
    \draw[->, blue!55!black, thick] (w2.south) ++(0.3,0) -- ++(0,-0.45)
      node[midway, right, note, text=blue!55!black] {append \texttt{e}, drop c$_2$};
    \node[box, fill=green!9, draw=green!45!black, anchor=west]
      at (0,-3*\dy-0.15)
      {$R(\texttt{tea})=\dfrac{\text{green letters}}{\text{letters}}
        =\dfrac{2}{3}$};
    \node[note, anchor=west] at (5.0,-3*\dy-0.15)
      {compared with the other candidates' $R$;\\
       top-$\gamma$ receive $+\delta$ on their logit};
    \draw[gray!50, dashed] (-1.2,-3*\dy-0.8) -- (10.6,-3*\dy-0.8);
    \def\yd{-3*\dy-1.8}
    \node[lab] at (-0.4,\yd) {detect};
    \foreach \i/\c/\s in {0/c$_1$/outw,1/c$_2$/outw,2/c$_3$/inw,3/c$_4$/inw,
                          4/t/inw,5/e/inw,6/a/hit,7/c$_5$/outw,8/c$_6$/outw}
      \node[\s, minimum height=6mm, minimum width=5.6mm, inner sep=0pt,
            rounded corners=0pt, font=\ttfamily\small] (d\i) at (\cw*\i,\yd) {\c};
    \node[win, fit=(d2)(d5)] (wd) {};
    \node[hash] (hd) at (6.0,\yd) {$\PRF_k(\cdot)$};
    \node[glist, right=of hd] (gd) {$\{a,s,i\}$};
    \node[note, right=2mm of gd] {\texttt{a} green, $+1$};
    \draw[->] (wd.north) -- ++(0,0.35) -| (hd.north);
    \draw[->] (hd) -- (gd);
  \end{tikzpicture}}
  \caption{LetterMark scoring the candidate token \texttt{tea} after
    context letters $c_1c_2c_3c_4$ (top) and detecting the mark (bottom).
    The blue window holds the last $h=4$ letters, and its keyed hash selects
    the green list for the next letter; after each check the letter is
    appended and the window slides by one. Two of the three letters are
    green, so $R(\texttt{tea})=2/3$. The verifier repeats the same scan on
    the published text without a tokenizer.}
  \label{fig:lettermark-example}
\end{figure}
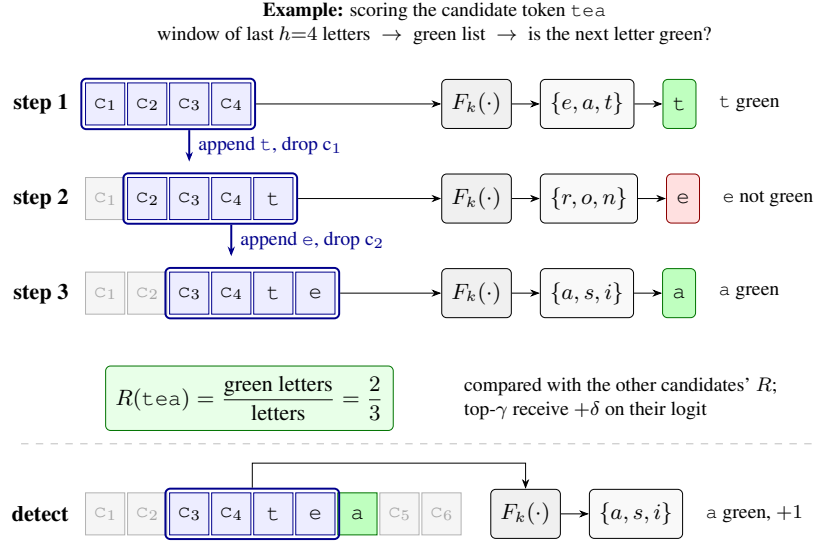

Letters can serve as the carrier because they have two properties:
\begin{itemize}[leftmargin=2em,nosep]
 \item \textbf{Publicly accessible.} The letters are part of the published
 text and can be recovered by any verifier without access to private
 artifacts.
 \item \textbf{Representation-independent.} The letter sequence is determined
 by the published text itself, rather than by a model-specific
 representation. The same text therefore provides the same carrier
 across models and over time.
\end{itemize}

\subsection{Embedding a watermark through the byte stream}
\label{sec:bytemark}

LetterMark illustrates a publicly accessible, representation-independent
carrier, but letters exist only in alphabetic scripts. General text needs a
carrier with the same two properties that also covers numerals, punctuation,
non-Latin scripts, and the other characters a text may contain. We propose
\ByteMark, which takes the UTF-8 byte stream of the text as its watermark
carrier. UTF-8 maps Unicode code points to byte sequences that a verifier can
recover from the published text without information from the generator.

For a secret key $k$, byte context $c$, and next byte $b$, define
\begin{equation}
 u_b(c)=\frac{(\operatorname{int}_{64}\circ\PRF_k)(c,b)}{2^{64}}\in[0,1),
 \label{eq:compact-prf}
\end{equation}
where $\operatorname{int}_{64}$ reads the first eight bytes of a PRF output as a big-endian unsigned integer, and $\PRF_k$ is a keyed pseudorandom function (PRF) under the secret key
$k$: a deterministic function whose outputs are indistinguishable from uniform random values
without the key, built from cryptographic primitives widely used in secure
communications \citep{rfc8446,rfc4231}.
The public specifications of these primitives support reproducible computation across platforms.

The context $c$ contains at most the last $h$ bytes of the generated completion.
For a candidate token $v$ with $L_v$ bytes, we score each byte $b_j$ using
its rolling context $c_j=\operatorname{last}_h(c\Vert b_1\Vert\cdots\Vert b_{j-1})$,
which includes earlier bytes of the
candidate token. Here $\operatorname{last}_h$ retains at most the final $h$
bytes, and $\Vert$ denotes concatenation. We then average the scores to bias the token logits:
\begin{equation}
 R(v;c)=\frac{1}{L_v}\sum_{j=1}^{L_v}u_{b_j}(c_j),\qquad
 \ell'_v=\ell_v+\delta\,\mathbf 1[v\in\operatorname{Top}_\gamma R].
 \label{eq:compact-generation}
\end{equation}
Here $\ell_v$ is the original logit, $\operatorname{Top}_\gamma R$ contains
the highest-scoring fraction $\gamma$ of candidates, and $\delta$ is the
logit increase. This bias encourages tokens with larger $R$, favoring higher
byte scores $u$ across the watermarked text. The model samples from
$\operatorname{softmax}(\ell')$.

Each generator applies the shared byte-level rule through its own token
vocabulary, while detection depends only on the text's bytes, key, and public
specification; Figure~\ref{fig:compact-interface} contrasts these verification inputs with those of the evaluated baselines. The byte-level carrier also supports other injection
rules: Appendix~\ref{sec:multilayer} outlines a possible extension to
SynthID-style multilayer reweighting. We use the single-layer green-set logit
bias for computational economy and implementation simplicity.

Because this specification is fixed and its computation deterministic,
the detection procedure does not change when models or their tokenizers are updated.

Algorithms~\ref{alg:gen} and~\ref{alg:detect} in Appendix~\ref{sec:algorithm}
give the generation and detection procedures, together with details on UTF-8, the PRF construction, and implementation
rules.

\begin{figure}[t]
  \centering
  \newcommand{\ifcmark}{\textcolor{green!45!black}{\ding{51}}}%
  \newcommand{\ifxmark}{\textcolor{red!70!black}{\ding{55}}}%
  \resizebox{.8\linewidth}{!}{%
  \begin{tikzpicture}[
      font=\footnotesize,
      box/.style={draw, rounded corners=1.5pt, align=center, inner sep=1.5pt,
                  minimum width=14mm, minimum height=5mm, font=\scriptsize},
      private/.style={box, fill=red!11, draw=red!65!black},
      public/.style={box, fill=blue!8, draw=blue!55!black},
      cell/.style={draw=gray!50, minimum width=14mm, minimum height=5.6mm},
      lab/.style={font=\footnotesize\bfseries, align=center},
    ]
    \node[align=center] at (2.9,1.35) {\textbf{Evaluated baselines}\\{\scriptsize detector reads with the tokenizer or embedding of}};
    \foreach \d/\x in {A/1.4,B/2.9,C/4.4}
      \node[private] at (\x,0.6) {generator \d};
    \node[align=center] at (6.6,1.35) {\textbf{\ByteMark}\\{\scriptsize\strut}};
    \node[public] at (6.6,0.6) {one detector};
    \foreach \s/\y in {A/0,B/-0.56,C/-1.12} {
      \node[anchor=east] at (0.62,\y) {generator \s};
      \foreach \d/\x in {A/1.4,B/2.9,C/4.4} {
        \ifx\s\d \node[cell, fill=green!10] at (\x,\y) {\ifcmark};
        \else \node[cell, fill=red!6] at (\x,\y) {\ifxmark}; \fi }
      \node[cell, fill=green!10] at (6.6,\y) {\ifcmark};
    }
    \draw[gray, dashed] (5.5,1.0) -- (5.5,-1.45);
  \end{tikzpicture}}
  \caption{Each baseline detector recovers the watermark only with the
    representation of the generating source, whereas one \ByteMark detector
    serves all three sources.}
  \label{fig:compact-interface}
\end{figure}

\subsection{Text-only detection}
\label{sec:text-only-detection}

To establish a theoretical basis for detection, we first examine the
distribution of byte scores under the null hypothesis: the text is
unwatermarked and independent of the watermark key. In the ideal
continuous-score model, distinct context--byte pairs receive independent
$U[0,1)$ scores; for a text fixed independently of these scores, all probabilities and expectations below are over this random assignment. Detection tests whether the mean of the text's distinct byte scores exceeds its
null expectation.

For a text $x$ with byte stream $b_1,\ldots,b_N$, the verifier reconstructs
$c_i=\operatorname{last}_h(b_1\Vert\cdots\Vert b_{i-1})$ and
$u_i=u_{b_i}(c_i)$. We retain each distinct $(c_i,b_i)$ pair once, following
the treatment of repeated $n$-grams in prior work
\citep{kirchenbauer2024reliability,fernandez2023three}. Let $m\geq1$ be the
number retained. Reindex their scores as $u_1,\ldots,u_m$ and let $\bar
u=m^{-1}\sum_{i=1}^m u_i$. The index $j$ above is local to a candidate token;
before deduplication, $i$ indexes the complete byte stream.

\paragraph{Null calibration.}
\begin{theorem}[Fixed-text null calibration]
\label{thm:compact-calibration}
Under the null hypothesis, the $m$ retained scores
are independent $U[0,1)$ values, so $\mathbb E[\bar u]=1/2$,
$\operatorname{Var}(\bar u)=1/(12m)$, and $m\bar u$ has the $m$-term
Irwin--Hall distribution with CDF
\begin{equation}
 F_m(t)=\frac{1}{m!}\sum_{j=0}^{\lfloor t\rfloor}(-1)^j\binom{m}{j}(t-j)^m,
 \qquad 0\leq t\leq m.
 \label{eq:compact-irwin-hall}
\end{equation}
Consequently, the standardized mean
\begin{equation}
 z(x)=\sqrt{12m}\left(\bar u-\frac{1}{2}\right)
 \label{eq:compact-detector}
\end{equation}
satisfies $\mathbb E[z]=0$ and $\operatorname{Var}(z)=1$, and for every $m\geq1$
and $0<\alpha<1$,
\begin{equation}
 \Pr[z>\tau(m,\alpha)]=\alpha,\qquad
 \tau(m,\alpha)=\frac{F_m^{-1}(1-\alpha)-m/2}{\sqrt{m/12}}.
 \label{eq:compact-exact-threshold}
\end{equation}
As $m\to\infty$,
\begin{equation}
 z(x)\xrightarrow{d}\mathcal N(0,1),\qquad
 \tau(m,\alpha)\to\Phi^{-1}(1-\alpha),
 \label{eq:compact-asymptotic}
\end{equation}
for each fixed $0<\alpha<1$, where $\Phi$ is the standard normal CDF.
\end{theorem}
\begin{proof}
See Appendix~\ref{sec:null-proof}.
\end{proof}

The detector therefore uses $z(x)$ as its statistic and reports watermark
evidence when $z(x)>\tau$ for a threshold $\tau$; a larger $z$ indicates a
higher average byte score than the null mean $1/2$. In our experiments, we use
empirically calibrated detection thresholds (Appendix~\ref{sec:thresholds}) so
that reported operating points reflect the evaluated human pools. A deployment
without calibration data can use the theoretical threshold;
Appendix~\ref{sec:null-realized} compares its nominal false-positive rates
with the rates realized on human text.

We next show that the injection raises the expected token score: at each decoding step, the boost cannot lower it and raises it strictly whenever the boosted
set has the higher conditional mean.
\begin{proposition}[One-step score increase]
\label{prop:one-step-score}
Fix the decoding history and candidate scores $R(v;c)$. Let
$p=\operatorname{softmax}(\ell)$ and $q=\operatorname{softmax}(\ell')$,
where $\ell'$ applies a boost $\delta\geq0$ to the highest-scoring
candidate set $G=\operatorname{Top}_\gamma R$.
Write $p_G=\sum_{v\in G}p(v)$ and $q_G=\sum_{v\in G}q(v)$.
For $0<p_G<1$, define the conditional means
$\mu_G=\mathbb E_p[R\mid v\in G]$ and
$\mu_{\bar G}=\mathbb E_p[R\mid v\notin G]$. Then
\begin{equation}
 \mathbb E_q[R]-\mathbb E_p[R]
 =(q_G-p_G)(\mu_G-\mu_{\bar G})\geq0.
 \label{eq:one-step-score-increase}
\end{equation}
The increase is strict if $\delta>0$ and $\mu_G>\mu_{\bar G}$.
\end{proposition}
\begin{proof}
See Appendix~\ref{sec:one-step-proof}.
\end{proof}

A
closed-form characterization of the document-level statistic under
watermarking requires additional assumptions on the byte distribution and
dependencies across decoding steps; we leave this analysis outside the scope
of this work.


\section{Experiments}
\label{sec:experiments}

We empirically evaluate representation-independent detection, continuity
across participating models, and generation efficiency.
\begin{itemize}[leftmargin=2em,nosep]
 \item \textbf{Detection without tokenizer or embedding dependencies
 (Section~\ref{sec:exp-replacement}).}
 Token-level and sentence-level detectors lose most of their signal when
 the detector's tokenizer or embedding model differs from the generator's,
 whereas \ByteMark requires neither.
 \item \textbf{Detection continuity across participating models
 (Section~\ref{sec:exp-composition}).}
 We evaluate one unchanged detector on text continued by participating
 models or rewritten by them.
 \item \textbf{Generation efficiency (Section~\ref{sec:exp-efficiency}).}
 We measure output throughput and the number of candidate tokens
 generated per final output token.
\end{itemize}

Appendix~\ref{sec:protocol} states the experimental protocol; the other
appendices report additional experiments, analyses, and diagnostics.

\subsection{Detection under embedding model and tokenizer replacement}
\label{sec:exp-replacement}

We evaluate the representation dependence discussed in
Section~\ref{sec:detection-requirements} by replacing the detector's
representation while holding the generated text fixed. \ByteMark uses neither
component, so its rows simply report detection on the three datasets.

\begin{table}[H]
\centering\small
\caption{TPR (\%) of the token-level detectors at the same-pool empirical 1\% threshold. Matched uses the generating tokenizer; each Replacement cell averages the five other tokenizers over the three generators of that dataset. \ByteMark uses no tokenizer, so its cells are its detection rates.}
\label{tab:tokenizer-domain-summary}
\label{tab:compact-mismatch-token}
\maxsizebox{\linewidth}{!}{%
\begin{tabular}{lcccc}
\toprule
Method & \multicolumn{3}{c}{Replacement tokenizer} & Matched \\
\cmidrule(lr){2-4}
 & C4 & Wikipedia & LFQA & \\
\midrule
\ifcsvstrcmp{ByteMark}{ByteMark}{\textbf{\ByteMark}}{ByteMark} & \ifcsvstrcmp{1}{1}{\textbf{99.5}}{99.5} & \ifcsvstrcmp{1}{1}{\textbf{99.3}}{99.3} & \ifcsvstrcmp{1}{1}{\textbf{99.6}}{99.6} & 99.5\\
\ifcsvstrcmp{KGW}{ByteMark}{\textbf{\ByteMark}}{KGW} & \ifcsvstrcmp{0}{1}{\textbf{1.5}}{1.5} & \ifcsvstrcmp{0}{1}{\textbf{1.3}}{1.3} & \ifcsvstrcmp{0}{1}{\textbf{1.7}}{1.7} & 99.6\\
\ifcsvstrcmp{SynthID}{ByteMark}{\textbf{\ByteMark}}{SynthID} & \ifcsvstrcmp{0}{1}{\textbf{0.6}}{0.6} & \ifcsvstrcmp{0}{1}{\textbf{0.6}}{0.6} & \ifcsvstrcmp{0}{1}{\textbf{0.6}}{0.6} & 100.0\\
\ifcsvstrcmp{SIR}{ByteMark}{\textbf{\ByteMark}}{SIR} & \ifcsvstrcmp{0}{1}{\textbf{1.7}}{1.7} & \ifcsvstrcmp{0}{1}{\textbf{5.9}}{5.9} & \ifcsvstrcmp{0}{1}{\textbf{2.2}}{2.2} & 96.5\\
\ifcsvstrcmp{EXP}{ByteMark}{\textbf{\ByteMark}}{EXP} & \ifcsvstrcmp{0}{1}{\textbf{1.6}}{1.6} & \ifcsvstrcmp{0}{1}{\textbf{1.2}}{1.2} & \ifcsvstrcmp{0}{1}{\textbf{1.3}}{1.3} & 99.9\\
\ifcsvstrcmp{ITS}{ByteMark}{\textbf{\ByteMark}}{ITS} & \ifcsvstrcmp{0}{1}{\textbf{0.8}}{0.8} & \ifcsvstrcmp{0}{1}{\textbf{0.8}}{0.8} & \ifcsvstrcmp{0}{1}{\textbf{1.0}}{1.0} & 98.0\\
\bottomrule
\end{tabular}
}
\end{table}

\begin{table}[H]
\centering\small
\caption{TPR (\%) of the sentence-level methods on C4 at the same-pool
empirical 1\% threshold. Matched uses the generating embedding model;
Replacement averages the replacement models that can score the method
(Appendix~\ref{sec:encoder-replacement}).}
\label{tab:e5}
\maxsizebox{\linewidth}{!}{%
\begin{tabular}{lrr}
\toprule
Method & Replacement & Matched \\
 & embedding & \\
\midrule
SimMark & 9.2 & 91.8\\
SemStamp & 3.8 & 99.1\\
k-SemStamp & 3.7 & 96.6\\
AliMark & 6.5 & 99.8\\
\bottomrule
\end{tabular}
}
\end{table}

\paragraph{Results.}
Table~\ref{tab:tokenizer-domain-summary} shows that every token-level
baseline detects its watermark with its own tokenizer, but detection largely fails with replacement tokenizers, whereas \ByteMark
retains its detection rate in every domain. The complete generator--tokenizer
matrices are in Tables~\ref{tab:e2}--\ref{tab:e02-its} of
Appendix~\ref{sec:core}.

Table~\ref{tab:e5} reports the sentence-level methods under embedding-model
replacement. With their own embedding model all four reach a TPR above 90\%;
with a replacement model the mean TPR falls to 4--9\%.
Appendix~\ref{sec:encoder-replacement} reports each replacement model
separately.

For token-level detectors, searching over candidate tokenizers does not remove the dependence on the generating tokenizer. For Table~\ref{tab:search-main}, a search flags a human text whenever any of $K$ candidate tokenizers exceeds its own same-pool threshold at level $\alpha$. Its false-positive rate is several times the single-test rate and close to $1-(1-\alpha)^K$, the value for independent tests, so the candidates' false-positive sets overlap little at this threshold and the combined decision must be calibrated. Bonferroni correction achieves this without assuming independence by testing each candidate at level $\alpha/K$ (Appendix~\ref{sec:search}). These stricter thresholds can reduce power relative to detection with the known tokenizer, and the search still requires the candidate set to contain the generating tokenizer.

\begin{table}[t]
\centering\small
\caption{Same-pool empirical false-positive rate of an uncorrected search that
flags a human text when any of $K=6$ candidate tokenizers exceeds its own
threshold at level $\alpha=1\%$. Independent candidates would give
$1-(1-\alpha)^K\approx5.85\%$; identical rejection events would give $\alpha$.
Bonferroni correction tests each candidate at $\alpha/K\approx0.167\%$.}
\label{tab:search-main}
\begin{tabular}{lcc}
\toprule
Method & Observed & Independent \\
\midrule
KGW & 0.0580 & 0.0585\\
SynthID & 0.0580 & 0.0585\\
SIR & 0.0537 & 0.0585\\
EXP & 0.0563 & 0.0585\\
ITS & 0.0573 & 0.0585\\
\bottomrule
\end{tabular}

\end{table}

\paragraph{Summary.}
Token-level detectors need the generating tokenizer and sentence-level detectors
need the generating embedding model; with a replacement tokenizer or embedding model, most of the
watermark signal is lost. \ByteMark depends
on neither, so tokenizer replacement leaves its detection rate in Table~\ref{tab:tokenizer-domain-summary} unchanged. This is the
model-independence property of Section~\ref{sec:introduction}, and it is what
unknown-source detection requires: given the text and the key, one detector
verifies text from any participating model without identifying its source. Searching over candidate tokenizers does not remove this requirement for the baselines: it must be calibrated for multiple testing and still needs the generating tokenizer among its candidates.

\subsection{Detection continuity across model transitions}
\label{sec:exp-composition}

We test whether one unchanged detector remains effective when participating
models continue a text or re-embed the same watermark during rewriting.

\paragraph{Relay generation.}
In this experiment, Pythia-1.4B, Qwen2.5-1.5B, and OpenLLaMA-3B write successive segments of a C4
completion in turn, each embedding the watermark under the same key
(Appendix~\ref{sec:composition}). Figure~\ref{fig:relay-signal} (left) plots
the mean local signal in a sliding window of five tokens along the relay. A KGW
or ITS detector reads the whole text with one fixed tokenizer, so it finds the
signal only in the segment written with that tokenizer, and its curve stays
near zero in the other two. The \ByteMark curve holds its level across all three
segments and shows no visible change at either switch: the byte context carries
over each model boundary, and one detector reads the text as a single byte stream.

\paragraph{Rewriting with re-embedding.}
In this experiment, Parrot, DIPPER, and Mistral-7B-Instruct rewrite watermarked
Pythia-1.4B completions, each embedding the same \ByteMark watermark as it
rewrites. The unchanged \ByteMark detector separates the rewritten outputs from
human text under all three rewriters (Figure~\ref{fig:relay-signal}, right).
This result concerns the re-embedded signal. Whether the original signal
survives rewriting by a non-participating model is a robustness question shared
by all watermarking methods, and Appendix~\ref{sec:rewriting} reports that
setting.

\paragraph{Summary.} A token-level detector reads only the part of a relayed document written with its own tokenizer, whereas one \ByteMark detector follows the watermark across model switches and after rewriting by participating models. This is the continuity-across-models property of Section~\ref{sec:introduction}.

\begin{figure}[t]
\centering
\begin{minipage}[c]{.6\linewidth}
\centering
\includegraphics[width=\linewidth]{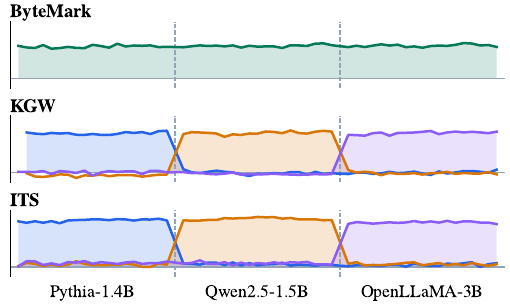}
\end{minipage}\hfill
\begin{minipage}[c]{.38\linewidth}
\centering\footnotesize
\setlength{\tabcolsep}{5pt}
\newcommand{\relaykey}[2]{\tikz[baseline=-.6ex]{\fill[#1, fill opacity=.5]
 (0,-.08) rectangle (.32,.12); \draw[#1, line width=.9pt] (0,.12) -- (.32,.12);}~#2}
\definecolor{relayB}{HTML}{047857}\definecolor{relayP}{HTML}{2563eb}%
\definecolor{relayQ}{HTML}{d97706}\definecolor{relayO}{HTML}{8b5cf6}%
\begin{tabular}{@{}l@{}}
\relaykey{relayB}{\ByteMark}\\
\relaykey{relayP}{Pythia tokenizer}\\
\relaykey{relayQ}{Qwen tokenizer}\\
\relaykey{relayO}{OpenLLaMA tokenizer}
\end{tabular}\par\vspace{10pt}
\ByteMark, participating rewriter\par\vspace{2pt}
\begin{tabular}{lc}
\toprule
Rewriter & AUROC \\
\midrule
Parrot & 0.9985\\
DIPPER & 0.9998\\
Mistral-7B & 0.9888\\
\bottomrule
\end{tabular}

\end{minipage}
\caption{Detection across model transitions. Left: mean local $z$ in a
sliding window of five tokens along the C4 relay; dashed lines mark model switches.
Right: \ByteMark AUROC after rewriting by participating models of Pythia-1.4B
completions.}
\label{fig:relay-signal}
\end{figure}

 \subsection{Generation efficiency}
\label{sec:exp-efficiency}

Sentence-level watermarks generate several candidates, keep one, and discard the rest (Section~\ref{sec:related}). We measure this cost in two ways: the number of candidate tokens generated per final output token, which counts the computation spent on discarded text and depends on the sampling mechanism rather than the implementation; and final-output tokens per second, the
speed a user observes.

Table~\ref{tab:main-generation-efficiency} reports both on the same requests.
The four sentence-level methods generate a large multiple of the tokens they
keep and reach only a small fraction of the unwatermarked throughput.
\ByteMark{} provides the contrast: it samples each token once, and its added
cost is scoring the vocabulary's bytes with a hardware-accelerated keyed
function at each step. A sentence-level watermark therefore
pays for its detection with generation efficiency, whereas \ByteMark{} keeps its throughput close to that of unwatermarked generation.

\begin{table}[H]
\centering
\small
\caption{Generation efficiency with OpenLLaMA-3B on 50 C4 texts.
Relative is the share of the unwatermarked rate (63.82 tok/s).
$^{*}$Different backend (vLLM).
Appendix~\ref{sec:generation-cost} gives the measurement conditions.}
\label{tab:main-generation-efficiency}
\begin{tabular}{lrrr}
\toprule
Method & Final tok/s & Relative (\%) & Generated / final \\
\midrule
Unwatermarked & 63.82 & 100.0 & 1.00\\
\ByteMark & 62.19 & 97.4 & 1.00\\
SemStamp & 1.73 & 2.7 & 35.38\\
k-SemStamp & 6.17 & 9.7 & 10.10\\
SimMark & 0.49 & 0.8 & 119.36\\
AliMark$^{*}$ & 18.22 & 28.6 & 100.68\\
\bottomrule
\end{tabular}

\end{table}

\section{Conclusion}
\label{sec:conclusion}
The token-level and sentence-level watermark tests we evaluate verify a presumed source: their detectors need the
source's tokenizer or embedding pipeline, which the text does not reveal, so a
negative result under one assumed source says little about whether a passage
is watermarked. \ByteMark removes this dependence by defining the watermark on
the UTF-8 bytes of the published text, so that detection needs only the text,
the key, and the public specification. In our experiments, replacing the
tokenizer or the embedding model leaves the token-level and sentence-level
baselines with little of their signal, whereas \ByteMark's detection rate is
unchanged; one detector follows a text across a relay of heterogeneous
generators and recovers the watermark that a participating rewriter re-embeds;
and sampling each output token once keeps throughput close to unwatermarked
generation, with an overhead that shrinks as models grow. A byte-level
watermark thus offers \textbf{unified detection} of text from any participating
model without identifying its source. Appendix~\ref{sec:discussion} discusses
the limitations.

\paragraph{Code availability.} The implementation and evaluation code will be released as open source after cleanup.

\par
\endgroup
\section*{Ethics Statement}
This work uses public corpora and public models and involves no human
subjects.  A watermark detector supports claims about the origin of a text,
and a false claim can harm the person it is attributed to.  We therefore
report detection at frozen thresholds calibrated on held-out human text and
state the realized false-positive rates, and we recommend that a positive
detection be treated as evidence rather than proof.  The key must remain
secret; its loss exposes the watermark of every participant, as discussed in
Appendix~\ref{sec:discussion}.

\section*{AI Use Statement}
Generative AI tools were used to summarize literature, analyze manuscript
structure, translate draft material, edit the manuscript for readability
and consistency, and assist with editing the experimental code.  We take
responsibility for the final content of this work, including text and artifacts
produced with the aid of generative AI.

\bibliographystyle{plainnat}
\bibliography{references}

\appendix
\newenvironment{story}{\par}{\par}
\newcommand{\sectionlead}[1]{\begin{story}#1\end{story}}
\newcommand{\matchedpair}[2]{\shortstack{#1\\#2}}
\newcommand{\replacementpair}[3]{%
  \ifdim #3pt<0.2pt\relax
    \cellcolor{yellow!20}\shortstack{#2\\#3}
  \else
    \shortstack{#2\\#3}%
  \fi}
\newcommand{\matchedcellslash}[1]{%
  \begin{tikzpicture}[baseline=0pt]
    \path[use as bounding box] (0,0) rectangle (8pt,15pt);
    \draw[line width=.4pt] (0,0) -- (4pt,7.5pt);
    \draw[line width=.4pt] (4pt,7.5pt) -- (8pt,15pt);
  \end{tikzpicture}%
}

\newcommand{\lowtprcell}[1]{%
  \ifcsvstrcmp{#1}{--}{#1}{%
    \ifdim #1pt<0.2pt\relax\cellcolor{yellow!20}\fi #1}%
}
\providecommand{\PRF}{F}   
\providecommand{\claudedel}[1]{\textcolor{red}{\sout{#1}}}
\providecommand{\claudeadd}[1]{\textcolor{red}{#1}}
\let\appendixSavedTable\table
\let\appendixSavedEndTable\endtable
\renewenvironment{table}[1][]{\appendixSavedTable[H]}{\appendixSavedEndTable}
\section*{Supplementary Material}
\phantomsection\addcontentsline{toc}{section}{Supplementary Material}

\section{Specification of byte-level generation and detection}
\label{sec:algorithm}
\sectionlead{The construction connects two interfaces: a generator chooses tokens,
while a verifier reads text. Both evaluate the same keyed function on the same
sequence of UTF-8 context--byte pairs. Token boundaries are needed to apply the
generation bias, but disappear from the detection procedure.}

\subsection{UTF-8 representation}
UTF-8 provides a standardized byte representation for Unicode text: it encodes
each Unicode scalar value in one to four bytes \citep{unicode16,rfc3629}. It is
also prevalent on the web: in W3Techs' 25 September 2026 survey, UTF-8 accounts
for 99.1\% of websites with an identified character encoding
\citep{w3techs2026encoding}. ByteMark uses the published text's UTF-8 byte stream
as a common carrier across generators, so a verifier can apply the same detection
rule without knowing the generator's token boundaries.

In our experiments, every text is produced directly as UTF-8 bytes through the
tokenizer's byte table, so no transcoding is involved. For text that a
deployment receives in another encoding, we recommend transcoding it to UTF-8
before detection. Text in another Unicode encoding, such as UTF-16 or UTF-32,
can be converted without additional Unicode normalization. Text in a legacy
character encoding, such as ISO-8859-1, Windows-1252, or Shift JIS, can be
decoded to Unicode through the published mapping of its declared encoding and
then re-encoded as UTF-8.

\subsection{Tokenizer coverage of the byte interface}
\label{sec:tokenizer-coverage}
This subsection checks whether diverse tokenizer vocabularies can supply the
exact output bytes needed to score candidate tokens during generation. For each
token ID $v$, the byte table must define $\beta(v)$, its contribution to the
UTF-8 output (Section~\ref{sec:byte-table}). We surveyed 35 public
causal-language-model tokenizers with pairwise distinct byte tables, including
byte-level BPE and SentencePiece-style vocabularies
(Table~\ref{tab:tokenizer-byte-coverage}). Every ID had a defined entry:
ordinary tokens contribute a nonempty byte sequence, while designated
zero-width special tokens contribute no output bytes.

\paragraph{UTF-8 byte-interface compatibility.}
A token ID $v$ is compatible if the byte table assigns it a determinate
output-byte sequence $\beta(v)$ under the tokenizer's decoding rules, without
replacing unresolved byte fragments with a replacement character. Let
$\mathcal V$ include all vocabulary IDs, including added and special tokens. We
measure compatibility as
\[
 \operatorname{Compatibility}(\mathcal V)
 =\frac{|\{v\in\mathcal V:\beta(v)\text{ is defined}\}|}{|\mathcal V|}.
\]
For example, OPT-1.3B's end-of-sequence token \texttt{</s>} maps to the empty
byte sequence and is therefore compatible. Further mapping details are given in
the reference implementation.
\begin{table}[tbp]
\centering\small
\caption{Token-to-byte mapping coverage for the 35 surveyed causal-language-model vocabularies. Every vocabulary ID has a defined byte contribution (100\% UTF-8 byte-interface compatibility).}
\label{tab:tokenizer-byte-coverage}
\setlength{\tabcolsep}{3pt}
\begin{tabular}{@{}lr@{\hspace{1em}}lr@{}}
\toprule
Tokenizer & IDs & Tokenizer & IDs \\
\midrule
xglm-564M & 256,008 & deepseek-llm-7b-base & 100,015 \\
bloom-560m & 250,680 & Jamba-v0.1 & 65,536 \\
MiniMax-M2 & 200,054 & falcon-7b & 65,024 \\
gpt-oss-120b & 200,019 & gpt-j-6b & 50,400 \\
Ling-1T & 157,153 & phi-2 & 50,295 \\
Seed-OSS-36B-Base & 155,121 & OLMoE-1B-7B-0924 & 50,280 \\
Qwen3-Coder-480B-A35B & 151,669 & gpt-neox-20b & 50,277 \\
Qwen2.5-72B & 151,665 & opt-1.3b & 50,265 \\
Qwen2-0.5B & 151,646 & gpt-neo-1.3B & 50,257 \\
GLM-4.5 & 151,365 & starcoder2-15b & 49,152 \\
Falcon3-1B-Base & 131,072 & granite-3.1-2b-base & 49,152 \\
DeepSeek-V3 & 128,815 & SmolLM2-1.7B & 49,152 \\
Nemotron-Super-49B & 128,256 & Mistral-Large-Instruct-2411 & 32,768 \\
EuroLLM-1.7B & 128,000 & Phi-3-mini-4k-instruct & 32,011 \\
EXAONE-3.5-2.4B-Instruct & 102,400 & Llama-2-7b & 32,000 \\
ERNIE-4.5-300B-A47B-Base-PT & 101,304 & Mixtral-8x22B & 32,000 \\
stablelm-2-1\_6b & 100,289 & open\_llama\_3b\_v2 & 32,000 \\
OLMo-2-32B & 100,278 &  &  \\
\bottomrule
\end{tabular}
\par\vspace{3pt}\footnotesize
A defined contribution may be a byte fragment or an empty special-token entry; individual tokens need not decode as complete UTF-8 characters. This is a mapping check, not a generation benchmark.
\end{table}

\subsection{Pseudorandom function}
\label{sec:keyed-hash}
ByteMark requires a keyed pseudorandom function (PRF) that maps each context--byte
input to a pseudorandom score; the function is not used to encrypt the text. The main
instantiation is AES-CMAC under a 128-bit key \citep{rfc4493}: AES is a
standardized block cipher, and CMAC turns it into a message authentication code
whose output serves as the PRF value. The same key and context--byte input
always reproduce the same score. Other standardized keyed PRFs can take
its place without changing the construction; the reference implementation also
supports keyed BLAKE3 \citep{oconnor2020blake3} and KMAC256
\citep{kelsey2016sp800185}.

\subsection{Detector key format}
\label{sec:detector-key-format}
A ByteMark key is a single string containing the construction, text encoding,
PRF, byte-context width, key material, and checksum, in that order. For the main
configuration, it has the form
\texttt{bytemark.utf8.aes-cmac.h8.SECRET.CHECK}. Under the fixed ByteMark
specification, this string determines the watermark's keyed byte-level signal.
A verifier given the published text and this string can compute the continuous
detection statistic without the generating model or its tokenizer.

\subsection{Byte table and candidate scores}
\label{sec:byte-table}
To score a candidate token, the generator needs its exact output-byte sequence
$\beta(v)=(b_1,\ldots,b_{L_v})$. The byte table recovers this sequence from the
tokenizer's representation, including byte-level BPE and SentencePiece
byte-fallback pieces. Decoding an isolated token can replace an incomplete UTF-8
fragment, so the table is not built from per-token \texttt{decode()} calls.
Special tokens omitted from the output map to the empty sequence; output-head
rows without token strings are masked.

Using the main-text notation, $u_b(c)$ is the keyed score of byte $b$ in
context $c$, with the fixed key $k$ implicit. For a candidate's $j$th byte, let
$c_j=\operatorname{last}_h(c\Vert b_1\Vert\cdots\Vert b_{j-1})$.
Here $\operatorname{last}_h$ retains at most the final $h$ bytes.
The candidate score, including the empty-sequence case, is
\begin{equation}
 R(v;c)=\begin{cases}
 \displaystyle\frac{1}{L_v}\sum_{j=1}^{L_v}u_{b_j}(c_j), & L_v>0,\\
 0, & L_v=0.
 \end{cases}
 \label{eq:candidate}
\end{equation}

\begin{algorithm}[H]
\caption{ByteMark generation through a token-based model}
\label{alg:gen}
\begin{algorithmic}[1]
\Require model $M$, byte table $\beta$, prompt $p$, key $k$, context width $h$,
 boosted fraction $\gamma$, logit increment $\delta$
\State Set $y$ and $c$ to the empty sequence
\While{the chosen stopping rule permits a step}
 \State $\ell\gets M(p\Vert y)$
 \ForAll{active vocabulary entries $v$}
 \State Compute $R(v;c)$ using Equation~\ref{eq:candidate}
 \EndFor
 \State $G\gets\operatorname{Top}_\gamma R$, the $\lfloor\gamma|\mathcal
 V|\rfloor$ highest-scoring entries
 \State $\ell'_v\gets\ell_v+\delta\,\mathbf 1[v\in G]$ for each $v$
 \State Sample $v$ from $\operatorname{softmax}(\ell')$; append $v$ to $y$
 \State $c\gets\operatorname{last}_h(c\Vert\beta(v))$
\EndWhile
\State Return the generated token sequence $y$
\end{algorithmic}
\end{algorithm}

\subsection{Detection from the completion bytes}
Detection scans the completion's UTF-8 encoding and retains one observation per
distinct context--byte pair. The rolling context advances at every byte, including
repeated pairs. For bytes $b_1,\ldots,b_N$, use the main-text definitions
$c_i=\operatorname{last}_h(b_1\Vert\cdots\Vert b_{i-1})$ and $u_i=u_{b_i}(c_i)$.
After retaining each distinct pair once, reindex the retained scores as
$u_1,\ldots,u_m$ and write their mean as $\bar u$.
The continuous statistic is
\begin{equation}
 z(x)=\sqrt{12m}\left(\bar u-\frac12\right).
 \label{eq:z}
\end{equation}

\begin{algorithm}[H]
\caption{ByteMark detection without a tokenizer or embedding model}
\label{alg:detect}
\begin{algorithmic}[1]
\Require nonempty text $x$, key $k$, context width $h$, threshold $\tau$
\State Set $c$ to the empty byte sequence; $D\gets\varnothing$; $S\gets0$; $m\gets0$
\ForAll{bytes $b$ in $\mathrm{UTF8}(x)$, in order}
 \If{$(c,b)\notin D$}
 \State Add $(c,b)$ to $D$; $S\gets S+u_b(c)$; $m\gets m+1$
 \EndIf
 \State $c\gets\operatorname{last}_h(c\Vert b)$
\EndFor
\State $\bar u\gets S/m$; $z(x)\gets\sqrt{12m}(\bar u-1/2)$
\State \Return $z(x)$ and the decision $z(x)>\tau$
\end{algorithmic}
\end{algorithm}

\section{Reconstruction and null calibration of the detection statistic}
\label{sec:theory}
\subsection{Invariance of the byte scores to token boundaries}
For an emitted token $v$ with $\beta(v)=(b_1,\ldots,b_{L_v})$, generation
scores its $j$th byte as $u_{b_j}(c_j)$, where
$c_j=\operatorname{last}_h(c\Vert b_1\Vert\cdots\Vert b_{j-1})$.
Let $i$ be this byte's position in the completed text $x$. Detection encounters
the same preceding bytes and therefore reconstructs $c_i=c_j$ and
$u_i=u_{b_j}(c_j)$. This equality holds for every emitted byte, regardless of
the token boundaries. Thus, the text, key, and public specification suffice to
reconstruct the byte scores.

\subsection{Effect of the logit bias}
The watermark signal is a shift in the keyed byte scores of the emitted text. At
a fixed decoding step, the logit bias shifts probability toward the highest-scoring candidates. We show below that this cannot decrease the expected candidate score.

The size of the shift at one step follows from the change in sampling
probabilities. Condition on the step's history and logits. As in the main text,
let $p=\operatorname{softmax}(\ell)$, $q=\operatorname{softmax}(\ell')$, and
$G=\operatorname{Top}_\gamma R$, and write $p_G=\sum_{v\in G}p(v)$ and
$q_G=\sum_{v\in G}q(v)$. At temperature one, adding $\delta$ to the logits in
$G$ multiplies their unnormalized softmax weights by $e^\delta$, so the
normalizing total becomes $(1-p_G)+e^\delta p_G=1+(e^\delta-1)p_G$. Hence
\begin{equation}
 q(v)=\frac{p(v)e^{\delta\mathbf1[v\in G]}}
 {1+(e^\delta-1)p_G},\qquad
 q_G-p_G=\frac{(e^\delta-1)p_G(1-p_G)}{1+(e^\delta-1)p_G}.
\end{equation}
This is the conditional probability transformation analyzed for logit-based
watermarking by \citet{wang2026powercalibrated}, applied here to
$G=\operatorname{Top}_\gamma R$.

\paragraph{Proof of Proposition~\ref{prop:one-step-score}.}
\label{sec:one-step-proof}
Fix the history, key, and logits at a generation step, so that $c$, all
scores $R(v;c)$, and $G$ are fixed before sampling. Write $R(v)=R(v;c)$ and
$\bar G=\mathcal V\setminus G$, and assume $\delta\geq0$. Expectations below
are over the sampled token only: $\E_s[R]=\sum_{v\in\mathcal V}s(v)R(v)$ for
$s\in\{p,q\}$. First suppose $0<p_G<1$, and define
\[
 \mu_G=\frac{\sum_{v\in G}p(v)R(v)}{p_G},\qquad
 \mu_{\bar G}=\frac{\sum_{v\in\bar G}p(v)R(v)}{1-p_G}.
\]
Let $Z=1+(e^\delta-1)p_G$. The preceding transformation gives
$q_G=e^\delta p_G/Z$ and $1-q_G=(1-p_G)/Z$. Hence
\[
 \frac{q(v)}{q_G}=\frac{p(v)}{p_G}\quad(v\in G),\qquad
 \frac{q(v)}{1-q_G}=\frac{p(v)}{1-p_G}\quad(v\in\bar G).
\]
Thus both conditional means are unchanged by the boost. Splitting each
expectation over the two sets yields
\begin{equation}
\begin{aligned}
 \E_q[R]-\E_p[R]
 &=\bigl[q_G\mu_G+(1-q_G)\mu_{\bar G}\bigr]
 -\bigl[p_G\mu_G+(1-p_G)\mu_{\bar G}\bigr]\\
 &=(q_G-p_G)(\mu_G-\mu_{\bar G})\\
 &=\frac{(e^\delta-1)p_G(1-p_G)}{Z}
 (\mu_G-\mu_{\bar G}).
\end{aligned}
\end{equation}
By the definition of $G=\operatorname{Top}_\gamma R$, every $v\in G$ and
$w\in\bar G$ satisfy $R(v)\geq R(w)$, including ties at the cutoff. Therefore
\[
 \mu_G-\mu_{\bar G}
 =\frac{\sum_{v\in G}\sum_{w\in\bar G}
 p(v)p(w)\bigl(R(v)-R(w)\bigr)}{p_G(1-p_G)}\geq0.
\]
Since $Z>0$ and $e^\delta-1\geq0$, it follows that
$\E_q[R]\geq\E_p[R]$. For $0<p_G<1$, the inequality is strict exactly when
$\delta>0$ and $\mu_G>\mu_{\bar G}$. If $p_G=0$ or $p_G=1$, the boost
leaves $q=p$, so the expected score is unchanged without requiring either
undefined conditional mean.

This conditional, one-step result does not determine the document-level
signal: candidate lengths, shared prefixes, repeated context--byte inputs, and
the contexts inherited by later steps all affect it, so we measure that signal
empirically.

Without the bias, every fixed nonempty candidate has expected score $1/2$
under an ideal PRF, and a candidate with $L_v$ distinct inputs has variance
$1/(12L_v)$. Averaging in Equation~\ref{eq:candidate}, rather than summing,
keeps the expected score from growing linearly with candidate length. The coverage diagnostics in Section~\ref{sec:diagnostics} examine which byte classes contribute this evidence.

\subsection{Proof of Theorem~\ref{thm:compact-calibration}}
\label{sec:null-proof}
\begin{proof}
\emph{Moments.} Fix a text independently of the ideal PRF, with $m>0$ distinct
context--byte inputs. Thus the retained $u_1,\ldots,u_m$ are independent
uniform values with expectation $1/2$ and variance $1/12$. Consequently,
\[
 \E[\bar u]=\frac12,\qquad \Var(\bar u)=\frac{1}{12m},\qquad
 \E[z]=0,\qquad \Var(z)=1.
\]

\emph{Exact distribution.} For a real argument $t\in[0,m]$, the CDF
$F_m(t)=\Pr[m\bar u\leq t]$ equals the volume of
$\{(u_1,\ldots,u_m)\in[0,1]^m:\sum_{i=1}^m u_i\leq t\}$.
Without the upper bounds $u_i\leq1$, this simplex has volume $t^m/m!$.
For a subset $J\subseteq\{1,\ldots,m\}$ of size $j$, requiring $u_i>1$
for every $i\in J$ and subtracting one from these coordinates gives volume
$\max(t-j,0)^m/m!$. Inclusion--exclusion over the $\binom mj$ subsets
of each size therefore yields
\begin{equation}
 F_m(t)=\frac{1}{m!}\sum_{j=0}^{\lfloor t\rfloor}(-1)^j\binom{m}{j}(t-j)^m,
 \qquad 0\leq t\leq m.
\end{equation}
\emph{Exact threshold.} Since
$z(x)=(m\bar u-m/2)/\sqrt{m/12}$, for a threshold $\tau$,
\[
 \Pr[z>\tau]=1-F_m\!\left(m/2+\tau\sqrt{m/12}\right).
\]
The CDF $F_m$ is continuous and strictly increasing on $[0,m]$.
Hence, for every $0<\alpha<1$, the ideal upper-tail level-$\alpha$ threshold is
\begin{equation}
 \tau(m,\alpha)=\frac{F_m^{-1}(1-\alpha)-m/2}{\sqrt{m/12}}.
\end{equation}
Substitution gives $\Pr[z>\tau(m,\alpha)]=\alpha$ exactly.

\emph{Asymptotic calibration.} The central limit theorem gives
$z(x)\xrightarrow{d}\mathcal N(0,1)$ as $m\to\infty$.
Because the standard normal CDF $\Phi$ is continuous and strictly increasing,
the corresponding quantiles converge:
$\tau(m,\alpha)\to\Phi^{-1}(1-\alpha)$ for each fixed $0<\alpha<1$.
\end{proof}

\subsection{Deduplication of repeated inputs}
A repeated context--byte pair reproduces exactly the same PRF score, rather than
an independent observation. Fix a text $x$ containing $N$ context--byte
occurrences and $m$ distinct pairs, where $m\leq N$. Let $u_1,\ldots,u_m$
be the distinct-pair scores, and let $A_i\geq1$ denote the multiplicity of
pair $i$, so that $\sum_{i=1}^m A_i=N$. Deduplication retains $m$ terms;
counting every occurrence retains $N$ terms, which can be grouped into $m$
terms weighted by $A_i$. Grouping repeated terms does not remove their
multiplicities.

\paragraph{Counting all occurrences with an independence assumption.}
Under the fixed-text ideal null of Theorem~\ref{thm:compact-calibration},
normalizing the sum over all $N$ occurrences as if they were independent gives
\[
 z_{\mathrm{naive}}(x)=\frac{\sum_{i=1}^m A_i(u_i-1/2)}{\sqrt{N/12}},
 \qquad
 \operatorname{Var}(z_{\mathrm{naive}})=\frac{\sum_{i=1}^m A_i^2}{N}.
\]
This variance exceeds one whenever any input repeats. Standard-normal
calibration is therefore unjustified and can understate the false-positive
rate. \citet[Appendix~A.3]{kirchenbauer2024reliability} report this discrepancy
for repeated token $n$-grams and show that excluding repeats improves calibration.
\citet[Sections~III.B--III.C]{fernandez2023three} address both finite-sample
tail approximation and repeated-input dependence, combining non-asymptotic
tests with exclusion of repeated context--token tuples.

\paragraph{Counting all occurrences with the correct null variance.}
Retaining the multiplicities and accounting for their covariance gives
\[
 z_{\mathrm{rep}}(x)=
 \frac{\sum_{i=1}^m A_i(u_i-1/2)}
 {\sqrt{\frac1{12}\sum_{i=1}^m A_i^2}}.
\]
For fixed multiplicities under the ideal null, its denominator is the standard
deviation of the numerator: the $A_i$ copies of $u_i$ are identical, so their
variance contribution is $A_i^2/12$, not $A_i/12$.
Unit variance alone does not make this weighted sum standard normal, so its
tail probabilities still require a suitable null distribution.

\paragraph{Counting each distinct input once.}
ByteMark removes the multiplicity weights and uses
\[
 z(x)=\frac{\sum_{i=1}^m(u_i-1/2)}{\sqrt{m/12}}.
\]
The denominator uses the number $m$ of distinct inputs because these $m$
scores are independent under the ideal null. Theorem~\ref{thm:compact-calibration}
gives the resulting exact null distribution; standard-normal calibration is
its large-$m$ approximation.

\paragraph{Standardized signal under a common shift.}
We compare $z(x)$ and $z_{\mathrm{rep}}(x)$ with the same $m$ distinct inputs
and fixed multiplicities $A_1,\ldots,A_m$, using their respective null standard
deviations. Assume a common-shift alternative
$\mathbb E[u_i]=1/2+\varepsilon$ for every distinct pair, with
$\varepsilon>0$ and the multiplicities held fixed. The expected scores,
measured in their respective null standard deviations, are
\[
 \mathbb E[z(x)]=\sqrt{12m}\,\varepsilon,\qquad
 \mathbb E[z_{\mathrm{rep}}(x)]=
 \sqrt{12}\,\varepsilon\,
 \frac{\sum_{i=1}^m A_i}{\sqrt{\sum_{i=1}^m A_i^2}}.
\]
Their ratio is therefore
\[
 \frac{\mathbb E[z_{\mathrm{rep}}(x)]}{\mathbb E[z(x)]}
 =\frac{\sum_{i=1}^m A_i}{\sqrt{m}\sqrt{\sum_{i=1}^m A_i^2}}\leq1,
\]
where Cauchy--Schwarz gives
$(\sum_i A_i)^2\leq m\sum_i A_i^2$, with equality exactly when all
multiplicities are equal. Thus, under the common-shift assumption, counting
repeated inputs provides no
benefit in standardized mean signal.

\section{Evaluation protocol}
\label{sec:protocol}

\paragraph{Generators, tokenizers, and corpora.}
The main experiments use Pythia-1.4B~\citep{biderman2023pythia},
Qwen2.5-1.5B~\citep{qwen2025qwen25}, and
OpenLLaMA-3B~\citep{geng2023openllama}, with their released tokenizers
spanning byte-level BPE and SentencePiece~\citep{kudo2018sentencepiece}.
The tokenizers of OPT-1.3B~\citep{zhang2022opt},
GPT-Neo-1.3B~\citep{black2021gptneo}, and
TinyLlama-1.1B~\citep{zhang2024tinyllama} are used only for detection.
Prompts and unwatermarked human texts come from C4~\citep{raffel2020t5},
Wikipedia articles in WikiText-103~\citep{merity2016wikitext}, and
LFQA~\citep{blagojevic2021lfqa}.

\paragraph{Baselines.}
Among token-level methods, we compare with the token-identity methods
KGW~\citep{kirchenbauer2023watermark}, EXP and
ITS~\citep{kuditipudi2023robust}, and
SynthID-Text~\citep{dathathri2024synthid}, as well as
SIR~\citep{liu2024semantic}, whose semantic features condition a
tokenizer-specific vocabulary mapping. The sentence-level baselines are
SemStamp~\citep{hou2024semstamp}, k-SemStamp~\citep{hou2024ksemstamp},
SimMark~\citep{li2025simmark}, and AliMark~\citep{li2026alimark}. We evaluate them under
representation mismatch by replacing the detector's native embedding model
while holding the generated text fixed.
Each baseline is run through its authors' released generation and detection
implementation. EXP and ITS follow the authors' conditional randomization procedure
\citep{kuditipudi2023robust}.

\subsection{Cohorts and evaluation lengths}
The token-level comparisons span C4, Wikipedia, and LFQA with Pythia-1.4B,
Qwen2.5-1.5B, and OpenLLaMA-3B. The semantic comparisons use C4 with the generation models of the original implementations and with Pythia-1.4B; some additional tests use Wikipedia with
Pythia-1.4B. In both families,
each source condition uses 500 generation requests and each domain supplies 1,000
shared human texts. Token-level generation uses a budget of 200 native tokens.
The matched-detection and tokenizer-replacement tables score prefixes of at most
200 reference tokens; shorter texts retain their original length. Tokenizer
replacement scores each generated text with its native tokenizer and with five
alternatives, without additional generation.

Detection metrics are computed over valid, scoreable texts. Paired analyses use the common
valid subset. Invalid entries, whether rejected generations or unscoreable
texts, are excluded and not replaced.

Sentence-level protocols retain their specified segmentation, candidate
generation, stopping rules, and timeout settings. Embedding-model replacement
retains the generated text and the method-specific watermark parameters;
incompatible configurations are reported as unscoreable rather than assigned
zero detection rates.

\subsection{Thresholds, true positive rates, and reference distributions}
\label{sec:thresholds}
For the displayed empirical comparisons, a detector's threshold is selected from
the shared human scores at a target 1\% exceedance rate. TPR is then measured on
scoreable watermarked positives. AUROC compares positive and human scores.

\subsection{Method configurations}
ByteMark uses AES-CMAC-128 with $h=8$ and the same $(\gamma,\delta)=(.25,2)$ as
KGW \citep{kirchenbauer2024reliability}; generation uses temperature one,
disabled top-$k$, and top-$p=1$. For both ByteMark and KGW, the logit increment
$\delta$ sets the watermark strength, and adjusting $\delta$ trades detection
rate against generation quality \citep{kirchenbauer2023watermark}.

\paragraph{Generation length.} ByteMark, KGW, and SynthID set
 \texttt{min\_new\_tokens} to the budget, which suppresses EOS before 200
 tokens. If the byte stream is incomplete at the budget, generation may continue
 for at most three additional tokens to complete a UTF-8 character; a completion
 that remains invalid UTF-8 is rejected without replacement. EXP and ITS run
 the authors' fixed number of sampling steps, where EOS neither stops nor is
 suppressed.

Semantic methods retain their author-specified preprocessing, generation, and
detection procedures
\citep{li2025simmark,hou2024semstamp,hou2024ksemstamp,li2026alimark}; only the
component named in each replacement experiment is changed. Structural
incompatibility, such as a replacement dimension inconsistent with stored
centroids, is recorded as undefined.

\subsection{Overview of the experiments}
\begin{center}\small
\begin{tabular}{@{}l>{\raggedright\arraybackslash}p{.36\linewidth}>{\raggedright\arraybackslash}p{.34\linewidth}@{}}
\toprule
Experiment & Varied component & Held fixed \\
\midrule
Matched detection & None (reference condition) & Generator, text, and detector \\
\midrule
Tokenizer replacement & Detector tokenizer or vocabulary mapping & Generated texts \\
\midrule
Candidate search & Detector over a finite candidate set & Human texts \\
\midrule
Embedding replacement & Detection embedding model & Generation and stored geometric artifacts \\
\midrule
Generator replacement & Generating model & Author watermark parameters \\
\midrule
Model relay & Generator at each relay stage & Watermark key and byte context \\
\midrule
Editing and rewriting & Text after generation & Clean-text detection threshold \\
\midrule
Statistical diagnostics & Length, statistic, character class, human corpus & Key and detector \\
\bottomrule
\end{tabular}
\end{center}

\section{Detection under representation replacement}
\label{sec:core}
\sectionlead{We first evaluate matched detection, then replace the detector's representation while keeping the generated text fixed. For token-level methods the verifier
replaces the tokenizer; for semantic methods it replaces the detection encoder. A
separate experiment changes the generation model instead, which separates dependence
on the generator from dependence on the detection representation. Search over
candidate tokenizers is treated separately in Section~\ref{sec:search}.}

\subsection{Matched detection}
Table~\ref{tab:e1} compares matched detection across source models and domains.
All methods exhibit strong detection performance under matched conditions
across the evaluated source models and domains.
\begin{table}[H]
\centering\small
\caption{Matched detection. Entries are AUROC / TPR.}
\label{tab:e1}
\maxsizebox{\linewidth}{!}{%
\begin{tabular}{llcccccc}
\toprule
Domain & Generator & \ByteMark & KGW & SynthID & SIR & EXP & ITS \\
\midrule
\ifcsvstrcmp{Pythia}{Pythia}{\ifnum1>1 \midrule\fi}{}\ifcsvstrcmp{Pythia}{Pythia}{C4}{} & Pythia & 0.9997 / 0.996 & 0.9998 / 0.996 & 0.9995 / 0.998 & 0.9993 / 0.984 & 0.9982 / 0.996 & 0.9911 / 0.980\\
\ifcsvstrcmp{Qwen}{Pythia}{\ifnum2>1 \midrule\fi}{}\ifcsvstrcmp{Qwen}{Pythia}{C4}{} & Qwen & 1.0000 / 0.998 & 1.0000 / 1.000 & 1.0000 / 1.000 & 0.9900 / 0.944 & 1.0000 / 1.000 & 0.9994 / 0.992\\
\ifcsvstrcmp{OpenLLaMA}{Pythia}{\ifnum3>1 \midrule\fi}{}\ifcsvstrcmp{OpenLLaMA}{Pythia}{C4}{} & OpenLLaMA & 0.9989 / 0.992 & 0.9998 / 0.994 & 1.0000 / 1.000 & 0.9991 / 0.986 & 0.9990 / 0.996 & 0.9944 / 0.980\\
\ifcsvstrcmp{Pythia}{Pythia}{\ifnum4>1 \midrule\fi}{}\ifcsvstrcmp{Pythia}{Pythia}{WIKI}{} & Pythia & 1.0000 / 1.000 & 1.0000 / 1.000 & 1.0000 / 1.000 & 0.9991 / 0.994 & 1.0000 / 1.000 & 0.9994 / 0.994\\
\ifcsvstrcmp{Qwen}{Pythia}{\ifnum5>1 \midrule\fi}{}\ifcsvstrcmp{Qwen}{Pythia}{WIKI}{} & Qwen & 0.9976 / 0.978 & 0.9994 / 0.984 & 1.0000 / 1.000 & 0.9887 / 0.914 & 1.0000 / 1.000 & 0.9877 / 0.928\\
\ifcsvstrcmp{OpenLLaMA}{Pythia}{\ifnum6>1 \midrule\fi}{}\ifcsvstrcmp{OpenLLaMA}{Pythia}{WIKI}{} & OpenLLaMA & 1.0000 / 1.000 & 0.9999 / 0.996 & 1.0000 / 1.000 & 0.9991 / 0.982 & 1.0000 / 1.000 & 0.9953 / 0.982\\
\ifcsvstrcmp{Pythia}{Pythia}{\ifnum7>1 \midrule\fi}{}\ifcsvstrcmp{Pythia}{Pythia}{LFQA}{} & Pythia & 1.0000 / 1.000 & 1.0000 / 1.000 & 1.0000 / 1.000 & 0.9987 / 0.986 & 0.9990 / 1.000 & 0.9987 / 0.994\\
\ifcsvstrcmp{Qwen}{Pythia}{\ifnum8>1 \midrule\fi}{}\ifcsvstrcmp{Qwen}{Pythia}{LFQA}{} & Qwen & 0.9996 / 0.990 & 0.9992 / 0.998 & 1.0000 / 1.000 & 0.9863 / 0.906 & 0.9998 / 0.998 & 0.9949 / 0.980\\
\ifcsvstrcmp{OpenLLaMA}{Pythia}{\ifnum9>1 \midrule\fi}{}\ifcsvstrcmp{OpenLLaMA}{Pythia}{LFQA}{} & OpenLLaMA & 0.9999 / 0.998 & 1.0000 / 1.000 & 1.0000 / 1.000 & 0.9995 / 0.988 & 0.9995 / 1.000 & 0.9982 / 0.990\\
\bottomrule
\end{tabular}
}
\end{table}

\subsection{Tokenizer replacement}
We hold the generated texts fixed and replace the detector's tokenizer or
vocabulary mapping. Under tokenizer mismatch, the other methods largely lose
their detection capability, whereas ByteMark's detection scores remain unchanged.
Tables~\ref{tab:e2}--\ref{tab:e02-its} report the replacement results, with
cells with TPR below $0.2$ shaded pale yellow.

SIR is included here because its detector uses a vocabulary mapping as well as
semantic conditioning. The model-matched comparison aligns each generator with
its native tokenizer before testing replacements.
\begingroup
\setlength{\intextsep}{8pt}
\begin{table}[H]
\centering\small
\caption{KGW: matched versus replacement tokenizers. Each cell shows AUROC (top) and TPR (bottom). Cells with TPR below $0.2$ are shaded pale yellow.}
\label{tab:e2}
\begingroup
\setlength{\tabcolsep}{3.5pt}
\renewcommand{\arraystretch}{1.0}
\maxsizebox{\linewidth}{!}{%
\begin{tabular}{@{}llc@{\hspace{8pt}}cccccc@{}}
\toprule
 & & \multicolumn{1}{c}{Matched} & \multicolumn{6}{c}{Replacement detector tokenizer} \\
\cmidrule(lr){3-3}\cmidrule(l){4-9}
Domain & Generator & Native & Pythia & Qwen & OpenLLaMA & OPT & GPT-Neo & TinyLLaMA \\
\midrule
\ifcsvstrcmp{0}{1}{\addlinespace[3pt]}{}C4 & Pythia & \matchedpair{0.9998}{0.996} & \ifcsvstrcmp{matched}{matched}{\matchedcellslash{Pythia}}{\replacementpair{}{matched}{}} & \ifcsvstrcmp{0.4560}{matched}{\matchedcellslash{Qwen}}{\replacementpair{}{0.4560}{0.016}} & \ifcsvstrcmp{0.5053}{matched}{\matchedcellslash{OpenLLaMA}}{\replacementpair{}{0.5053}{0.010}} & \ifcsvstrcmp{0.4954}{matched}{\matchedcellslash{OPT}}{\replacementpair{}{0.4954}{0.020}} & \ifcsvstrcmp{0.4181}{matched}{\matchedcellslash{GPT-Neo}}{\replacementpair{}{0.4181}{0.006}} & \ifcsvstrcmp{0.4368}{matched}{\matchedcellslash{TinyLLaMA}}{\replacementpair{}{0.4368}{0.012}}\\\addlinespace[0pt]
\ifcsvstrcmp{0}{1}{\addlinespace[3pt]}{} & Qwen & \matchedpair{1.0000}{1.000} & \ifcsvstrcmp{0.4591}{matched}{\matchedcellslash{Pythia}}{\replacementpair{}{0.4591}{0.006}} & \ifcsvstrcmp{matched}{matched}{\matchedcellslash{Qwen}}{\replacementpair{}{matched}{}} & \ifcsvstrcmp{0.4385}{matched}{\matchedcellslash{OpenLLaMA}}{\replacementpair{}{0.4385}{0.010}} & \ifcsvstrcmp{0.5373}{matched}{\matchedcellslash{OPT}}{\replacementpair{}{0.5373}{0.032}} & \ifcsvstrcmp{0.4868}{matched}{\matchedcellslash{GPT-Neo}}{\replacementpair{}{0.4868}{0.012}} & \ifcsvstrcmp{0.4901}{matched}{\matchedcellslash{TinyLLaMA}}{\replacementpair{}{0.4901}{0.014}}\\\addlinespace[0pt]
\ifcsvstrcmp{0}{1}{\addlinespace[3pt]}{} & OpenLLaMA & \matchedpair{0.9998}{0.994} & \ifcsvstrcmp{0.5164}{matched}{\matchedcellslash{Pythia}}{\replacementpair{}{0.5164}{0.008}} & \ifcsvstrcmp{0.5127}{matched}{\matchedcellslash{Qwen}}{\replacementpair{}{0.5127}{0.016}} & \ifcsvstrcmp{matched}{matched}{\matchedcellslash{OpenLLaMA}}{\replacementpair{}{matched}{}} & \ifcsvstrcmp{0.5766}{matched}{\matchedcellslash{OPT}}{\replacementpair{}{0.5766}{0.038}} & \ifcsvstrcmp{0.4702}{matched}{\matchedcellslash{GPT-Neo}}{\replacementpair{}{0.4702}{0.018}} & \ifcsvstrcmp{0.4466}{matched}{\matchedcellslash{TinyLLaMA}}{\replacementpair{}{0.4466}{0.012}}\\\addlinespace[0pt]
\ifcsvstrcmp{1}{1}{\addlinespace[3pt]}{}Wiki & Pythia & \matchedpair{1.0000}{1.000} & \ifcsvstrcmp{matched}{matched}{\matchedcellslash{Pythia}}{\replacementpair{}{matched}{}} & \ifcsvstrcmp{0.4048}{matched}{\matchedcellslash{Qwen}}{\replacementpair{}{0.4048}{0.006}} & \ifcsvstrcmp{0.3132}{matched}{\matchedcellslash{OpenLLaMA}}{\replacementpair{}{0.3132}{0.006}} & \ifcsvstrcmp{0.4135}{matched}{\matchedcellslash{OPT}}{\replacementpair{}{0.4135}{0.014}} & \ifcsvstrcmp{0.3410}{matched}{\matchedcellslash{GPT-Neo}}{\replacementpair{}{0.3410}{0.006}} & \ifcsvstrcmp{0.4224}{matched}{\matchedcellslash{TinyLLaMA}}{\replacementpair{}{0.4224}{0.002}}\\\addlinespace[0pt]
\ifcsvstrcmp{0}{1}{\addlinespace[3pt]}{} & Qwen & \matchedpair{0.9994}{0.984} & \ifcsvstrcmp{0.4955}{matched}{\matchedcellslash{Pythia}}{\replacementpair{}{0.4955}{0.010}} & \ifcsvstrcmp{matched}{matched}{\matchedcellslash{Qwen}}{\replacementpair{}{matched}{}} & \ifcsvstrcmp{0.2649}{matched}{\matchedcellslash{OpenLLaMA}}{\replacementpair{}{0.2649}{0.002}} & \ifcsvstrcmp{0.5724}{matched}{\matchedcellslash{OPT}}{\replacementpair{}{0.5724}{0.030}} & \ifcsvstrcmp{0.3625}{matched}{\matchedcellslash{GPT-Neo}}{\replacementpair{}{0.3625}{0.008}} & \ifcsvstrcmp{0.4321}{matched}{\matchedcellslash{TinyLLaMA}}{\replacementpair{}{0.4321}{0.004}}\\\addlinespace[0pt]
\ifcsvstrcmp{0}{1}{\addlinespace[3pt]}{} & OpenLLaMA & \matchedpair{0.9999}{0.996} & \ifcsvstrcmp{0.4134}{matched}{\matchedcellslash{Pythia}}{\replacementpair{}{0.4134}{0.008}} & \ifcsvstrcmp{0.3885}{matched}{\matchedcellslash{Qwen}}{\replacementpair{}{0.3885}{0.004}} & \ifcsvstrcmp{matched}{matched}{\matchedcellslash{OpenLLaMA}}{\replacementpair{}{matched}{}} & \ifcsvstrcmp{0.5755}{matched}{\matchedcellslash{OPT}}{\replacementpair{}{0.5755}{0.054}} & \ifcsvstrcmp{0.5505}{matched}{\matchedcellslash{GPT-Neo}}{\replacementpair{}{0.5505}{0.042}} & \ifcsvstrcmp{0.4150}{matched}{\matchedcellslash{TinyLLaMA}}{\replacementpair{}{0.4150}{0.004}}\\\addlinespace[0pt]
\ifcsvstrcmp{1}{1}{\addlinespace[3pt]}{}LFQA & Pythia & \matchedpair{1.0000}{1.000} & \ifcsvstrcmp{matched}{matched}{\matchedcellslash{Pythia}}{\replacementpair{}{matched}{}} & \ifcsvstrcmp{0.4772}{matched}{\matchedcellslash{Qwen}}{\replacementpair{}{0.4772}{0.016}} & \ifcsvstrcmp{0.5465}{matched}{\matchedcellslash{OpenLLaMA}}{\replacementpair{}{0.5465}{0.024}} & \ifcsvstrcmp{0.5427}{matched}{\matchedcellslash{OPT}}{\replacementpair{}{0.5427}{0.034}} & \ifcsvstrcmp{0.4128}{matched}{\matchedcellslash{GPT-Neo}}{\replacementpair{}{0.4128}{0.002}} & \ifcsvstrcmp{0.3994}{matched}{\matchedcellslash{TinyLLaMA}}{\replacementpair{}{0.3994}{0.004}}\\\addlinespace[0pt]
\ifcsvstrcmp{0}{1}{\addlinespace[3pt]}{} & Qwen & \matchedpair{0.9992}{0.998} & \ifcsvstrcmp{0.5296}{matched}{\matchedcellslash{Pythia}}{\replacementpair{}{0.5296}{0.025}} & \ifcsvstrcmp{matched}{matched}{\matchedcellslash{Qwen}}{\replacementpair{}{matched}{}} & \ifcsvstrcmp{0.4546}{matched}{\matchedcellslash{OpenLLaMA}}{\replacementpair{}{0.4546}{0.004}} & \ifcsvstrcmp{0.5696}{matched}{\matchedcellslash{OPT}}{\replacementpair{}{0.5696}{0.025}} & \ifcsvstrcmp{0.4360}{matched}{\matchedcellslash{GPT-Neo}}{\replacementpair{}{0.4360}{0.008}} & \ifcsvstrcmp{0.4419}{matched}{\matchedcellslash{TinyLLaMA}}{\replacementpair{}{0.4419}{0.006}}\\\addlinespace[0pt]
\ifcsvstrcmp{0}{1}{\addlinespace[3pt]}{} & OpenLLaMA & \matchedpair{1.0000}{1.000} & \ifcsvstrcmp{0.5844}{matched}{\matchedcellslash{Pythia}}{\replacementpair{}{0.5844}{0.030}} & \ifcsvstrcmp{0.5139}{matched}{\matchedcellslash{Qwen}}{\replacementpair{}{0.5139}{0.012}} & \ifcsvstrcmp{matched}{matched}{\matchedcellslash{OpenLLaMA}}{\replacementpair{}{matched}{}} & \ifcsvstrcmp{0.6097}{matched}{\matchedcellslash{OPT}}{\replacementpair{}{0.6097}{0.048}} & \ifcsvstrcmp{0.4590}{matched}{\matchedcellslash{GPT-Neo}}{\replacementpair{}{0.4590}{0.012}} & \ifcsvstrcmp{0.4025}{matched}{\matchedcellslash{TinyLLaMA}}{\replacementpair{}{0.4025}{0.008}}\\\addlinespace[0pt]
\bottomrule
\end{tabular}
}
\endgroup
\par\smallskip\begin{minipage}{\linewidth}\footnotesize
The matched reference is shown in the Native column; a diagonal stroke marks that same tokenizer, not a missing result.
\end{minipage}
\end{table}

\begin{table}[H]
\centering\small
\caption{SynthID: matched versus replacement tokenizers. Each cell shows AUROC (top) and TPR (bottom). Cells with TPR below $0.2$ are shaded pale yellow.}
\label{tab:e02-synthid-200}
\begingroup
\setlength{\tabcolsep}{3.5pt}
\renewcommand{\arraystretch}{1.0}
\maxsizebox{\linewidth}{!}{%
\begin{tabular}{@{}llc@{\hspace{8pt}}cccccc@{}}
\toprule
 & & \multicolumn{1}{c}{Matched} & \multicolumn{6}{c}{Replacement detector tokenizer} \\
\cmidrule(lr){3-3}\cmidrule(l){4-9}
Domain & Generator & Native & Pythia & Qwen & OpenLLaMA & OPT & GPT-Neo & TinyLLaMA \\
\midrule
\ifcsvstrcmp{0}{1}{\addlinespace[3pt]}{}C4 & Pythia & \matchedpair{0.9995}{0.998} & \ifcsvstrcmp{matched}{matched}{\matchedcellslash{Pythia}}{\replacementpair{}{matched}{}} & \ifcsvstrcmp{0.5134}{matched}{\matchedcellslash{Qwen}}{\replacementpair{}{0.5134}{0.002}} & \ifcsvstrcmp{0.5063}{matched}{\matchedcellslash{OpenLLaMA}}{\replacementpair{}{0.5063}{0.002}} & \ifcsvstrcmp{0.4874}{matched}{\matchedcellslash{OPT}}{\replacementpair{}{0.4874}{0.000}} & \ifcsvstrcmp{0.5054}{matched}{\matchedcellslash{GPT-Neo}}{\replacementpair{}{0.5054}{0.006}} & \ifcsvstrcmp{0.4857}{matched}{\matchedcellslash{TinyLLaMA}}{\replacementpair{}{0.4857}{0.008}}\\\addlinespace[0pt]
\ifcsvstrcmp{0}{1}{\addlinespace[3pt]}{} & Qwen & \matchedpair{1.0000}{1.000} & \ifcsvstrcmp{0.5045}{matched}{\matchedcellslash{Pythia}}{\replacementpair{}{0.5045}{0.008}} & \ifcsvstrcmp{matched}{matched}{\matchedcellslash{Qwen}}{\replacementpair{}{matched}{}} & \ifcsvstrcmp{0.5237}{matched}{\matchedcellslash{OpenLLaMA}}{\replacementpair{}{0.5237}{0.012}} & \ifcsvstrcmp{0.4874}{matched}{\matchedcellslash{OPT}}{\replacementpair{}{0.4874}{0.000}} & \ifcsvstrcmp{0.4845}{matched}{\matchedcellslash{GPT-Neo}}{\replacementpair{}{0.4845}{0.010}} & \ifcsvstrcmp{0.4814}{matched}{\matchedcellslash{TinyLLaMA}}{\replacementpair{}{0.4814}{0.004}}\\\addlinespace[0pt]
\ifcsvstrcmp{0}{1}{\addlinespace[3pt]}{} & OpenLLaMA & \matchedpair{1.0000}{1.000} & \ifcsvstrcmp{0.4914}{matched}{\matchedcellslash{Pythia}}{\replacementpair{}{0.4914}{0.010}} & \ifcsvstrcmp{0.4866}{matched}{\matchedcellslash{Qwen}}{\replacementpair{}{0.4866}{0.002}} & \ifcsvstrcmp{matched}{matched}{\matchedcellslash{OpenLLaMA}}{\replacementpair{}{matched}{}} & \ifcsvstrcmp{0.5015}{matched}{\matchedcellslash{OPT}}{\replacementpair{}{0.5015}{0.000}} & \ifcsvstrcmp{0.4813}{matched}{\matchedcellslash{GPT-Neo}}{\replacementpair{}{0.4813}{0.006}} & \ifcsvstrcmp{0.5121}{matched}{\matchedcellslash{TinyLLaMA}}{\replacementpair{}{0.5121}{0.020}}\\\addlinespace[0pt]
\ifcsvstrcmp{1}{1}{\addlinespace[3pt]}{}Wiki & Pythia & \matchedpair{1.0000}{1.000} & \ifcsvstrcmp{matched}{matched}{\matchedcellslash{Pythia}}{\replacementpair{}{matched}{}} & \ifcsvstrcmp{0.5100}{matched}{\matchedcellslash{Qwen}}{\replacementpair{}{0.5100}{0.004}} & \ifcsvstrcmp{0.4919}{matched}{\matchedcellslash{OpenLLaMA}}{\replacementpair{}{0.4919}{0.008}} & \ifcsvstrcmp{0.5013}{matched}{\matchedcellslash{OPT}}{\replacementpair{}{0.5013}{0.010}} & \ifcsvstrcmp{0.5200}{matched}{\matchedcellslash{GPT-Neo}}{\replacementpair{}{0.5200}{0.004}} & \ifcsvstrcmp{0.5113}{matched}{\matchedcellslash{TinyLLaMA}}{\replacementpair{}{0.5113}{0.004}}\\\addlinespace[0pt]
\ifcsvstrcmp{0}{1}{\addlinespace[3pt]}{} & Qwen & \matchedpair{1.0000}{1.000} & \ifcsvstrcmp{0.5103}{matched}{\matchedcellslash{Pythia}}{\replacementpair{}{0.5103}{0.006}} & \ifcsvstrcmp{matched}{matched}{\matchedcellslash{Qwen}}{\replacementpair{}{matched}{}} & \ifcsvstrcmp{0.5016}{matched}{\matchedcellslash{OpenLLaMA}}{\replacementpair{}{0.5016}{0.008}} & \ifcsvstrcmp{0.4846}{matched}{\matchedcellslash{OPT}}{\replacementpair{}{0.4846}{0.006}} & \ifcsvstrcmp{0.5174}{matched}{\matchedcellslash{GPT-Neo}}{\replacementpair{}{0.5174}{0.008}} & \ifcsvstrcmp{0.5201}{matched}{\matchedcellslash{TinyLLaMA}}{\replacementpair{}{0.5201}{0.006}}\\\addlinespace[0pt]
\ifcsvstrcmp{0}{1}{\addlinespace[3pt]}{} & OpenLLaMA & \matchedpair{1.0000}{1.000} & \ifcsvstrcmp{0.5147}{matched}{\matchedcellslash{Pythia}}{\replacementpair{}{0.5147}{0.004}} & \ifcsvstrcmp{0.5035}{matched}{\matchedcellslash{Qwen}}{\replacementpair{}{0.5035}{0.000}} & \ifcsvstrcmp{matched}{matched}{\matchedcellslash{OpenLLaMA}}{\replacementpair{}{matched}{}} & \ifcsvstrcmp{0.4971}{matched}{\matchedcellslash{OPT}}{\replacementpair{}{0.4971}{0.006}} & \ifcsvstrcmp{0.4960}{matched}{\matchedcellslash{GPT-Neo}}{\replacementpair{}{0.4960}{0.000}} & \ifcsvstrcmp{0.5079}{matched}{\matchedcellslash{TinyLLaMA}}{\replacementpair{}{0.5079}{0.010}}\\\addlinespace[0pt]
\ifcsvstrcmp{1}{1}{\addlinespace[3pt]}{}LFQA & Pythia & \matchedpair{1.0000}{1.000} & \ifcsvstrcmp{matched}{matched}{\matchedcellslash{Pythia}}{\replacementpair{}{matched}{}} & \ifcsvstrcmp{0.4869}{matched}{\matchedcellslash{Qwen}}{\replacementpair{}{0.4869}{0.008}} & \ifcsvstrcmp{0.5215}{matched}{\matchedcellslash{OpenLLaMA}}{\replacementpair{}{0.5215}{0.004}} & \ifcsvstrcmp{0.4985}{matched}{\matchedcellslash{OPT}}{\replacementpair{}{0.4985}{0.006}} & \ifcsvstrcmp{0.4529}{matched}{\matchedcellslash{GPT-Neo}}{\replacementpair{}{0.4529}{0.000}} & \ifcsvstrcmp{0.4994}{matched}{\matchedcellslash{TinyLLaMA}}{\replacementpair{}{0.4994}{0.000}}\\\addlinespace[0pt]
\ifcsvstrcmp{0}{1}{\addlinespace[3pt]}{} & Qwen & \matchedpair{1.0000}{1.000} & \ifcsvstrcmp{0.4706}{matched}{\matchedcellslash{Pythia}}{\replacementpair{}{0.4706}{0.004}} & \ifcsvstrcmp{matched}{matched}{\matchedcellslash{Qwen}}{\replacementpair{}{matched}{}} & \ifcsvstrcmp{0.4820}{matched}{\matchedcellslash{OpenLLaMA}}{\replacementpair{}{0.4820}{0.006}} & \ifcsvstrcmp{0.4960}{matched}{\matchedcellslash{OPT}}{\replacementpair{}{0.4960}{0.010}} & \ifcsvstrcmp{0.4618}{matched}{\matchedcellslash{GPT-Neo}}{\replacementpair{}{0.4618}{0.006}} & \ifcsvstrcmp{0.5208}{matched}{\matchedcellslash{TinyLLaMA}}{\replacementpair{}{0.5208}{0.002}}\\\addlinespace[0pt]
\ifcsvstrcmp{0}{1}{\addlinespace[3pt]}{} & OpenLLaMA & \matchedpair{1.0000}{1.000} & \ifcsvstrcmp{0.5077}{matched}{\matchedcellslash{Pythia}}{\replacementpair{}{0.5077}{0.018}} & \ifcsvstrcmp{0.4875}{matched}{\matchedcellslash{Qwen}}{\replacementpair{}{0.4875}{0.002}} & \ifcsvstrcmp{matched}{matched}{\matchedcellslash{OpenLLaMA}}{\replacementpair{}{matched}{}} & \ifcsvstrcmp{0.4925}{matched}{\matchedcellslash{OPT}}{\replacementpair{}{0.4925}{0.020}} & \ifcsvstrcmp{0.4679}{matched}{\matchedcellslash{GPT-Neo}}{\replacementpair{}{0.4679}{0.006}} & \ifcsvstrcmp{0.4948}{matched}{\matchedcellslash{TinyLLaMA}}{\replacementpair{}{0.4948}{0.000}}\\\addlinespace[0pt]
\bottomrule
\end{tabular}
}
\endgroup
\par\smallskip\begin{minipage}{\linewidth}\footnotesize
The matched reference is shown in the Native column; a diagonal stroke marks that same tokenizer, not a missing result.
\end{minipage}
\end{table}

\begin{table}[H]
\centering\small
\caption{SIR: matched versus replacement tokenizers. Each cell shows AUROC (top) and TPR (bottom). Cells with TPR below $0.2$ are shaded pale yellow.}
\label{tab:e02-sir-model-matched-c4}
\label{tab:e02-sir-model-matched-wiki}
\label{tab:e02-sir-model-matched-lfqa}
\begingroup
\setlength{\tabcolsep}{3.5pt}
\renewcommand{\arraystretch}{1.0}
\maxsizebox{\linewidth}{!}{%
\begin{tabular}{@{}llc@{\hspace{8pt}}cccccc@{}}
\toprule
 & & \multicolumn{1}{c}{Matched} & \multicolumn{6}{c}{Replacement detector tokenizer} \\
\cmidrule(lr){3-3}\cmidrule(l){4-9}
Domain & Generator & Native & Pythia & Qwen & OpenLLaMA & OPT & GPT-Neo & TinyLLaMA \\
\midrule
\ifcsvstrcmp{0}{1}{\addlinespace[3pt]}{}C4 & Pythia & \matchedpair{0.9993}{0.984} & \ifcsvstrcmp{matched}{matched}{\matchedcellslash{Pythia}}{\replacementpair{}{matched}{}} & \ifcsvstrcmp{0.4772}{matched}{\matchedcellslash{Qwen}}{\replacementpair{}{0.4772}{0.014}} & \ifcsvstrcmp{0.4222}{matched}{\matchedcellslash{OpenLLaMA}}{\replacementpair{}{0.4222}{0.004}} & \ifcsvstrcmp{0.5239}{matched}{\matchedcellslash{OPT}}{\replacementpair{}{0.5239}{0.018}} & \ifcsvstrcmp{0.5181}{matched}{\matchedcellslash{GPT-Neo}}{\replacementpair{}{0.5181}{0.024}} & \ifcsvstrcmp{0.6779}{matched}{\matchedcellslash{TinyLLaMA}}{\replacementpair{}{0.6779}{0.050}}\\\addlinespace[0pt]
\ifcsvstrcmp{0}{1}{\addlinespace[3pt]}{} & Qwen & \matchedpair{0.9900}{0.944} & \ifcsvstrcmp{0.4178}{matched}{\matchedcellslash{Pythia}}{\replacementpair{}{0.4178}{0.004}} & \ifcsvstrcmp{matched}{matched}{\matchedcellslash{Qwen}}{\replacementpair{}{matched}{}} & \ifcsvstrcmp{0.4782}{matched}{\matchedcellslash{OpenLLaMA}}{\replacementpair{}{0.4782}{0.004}} & \ifcsvstrcmp{0.3649}{matched}{\matchedcellslash{OPT}}{\replacementpair{}{0.3649}{0.002}} & \ifcsvstrcmp{0.5204}{matched}{\matchedcellslash{GPT-Neo}}{\replacementpair{}{0.5204}{0.014}} & \ifcsvstrcmp{0.4417}{matched}{\matchedcellslash{TinyLLaMA}}{\replacementpair{}{0.4417}{0.000}}\\\addlinespace[0pt]
\ifcsvstrcmp{0}{1}{\addlinespace[3pt]}{} & OpenLLaMA & \matchedpair{0.9991}{0.986} & \ifcsvstrcmp{0.3171}{matched}{\matchedcellslash{Pythia}}{\replacementpair{}{0.3171}{0.006}} & \ifcsvstrcmp{0.4777}{matched}{\matchedcellslash{Qwen}}{\replacementpair{}{0.4777}{0.042}} & \ifcsvstrcmp{matched}{matched}{\matchedcellslash{OpenLLaMA}}{\replacementpair{}{matched}{}} & \ifcsvstrcmp{0.5013}{matched}{\matchedcellslash{OPT}}{\replacementpair{}{0.5013}{0.008}} & \ifcsvstrcmp{0.5469}{matched}{\matchedcellslash{GPT-Neo}}{\replacementpair{}{0.5469}{0.042}} & \ifcsvstrcmp{0.4942}{matched}{\matchedcellslash{TinyLLaMA}}{\replacementpair{}{0.4942}{0.018}}\\\addlinespace[0pt]
\ifcsvstrcmp{1}{1}{\addlinespace[3pt]}{}Wiki & Pythia & \matchedpair{0.9991}{0.994} & \ifcsvstrcmp{matched}{matched}{\matchedcellslash{Pythia}}{\replacementpair{}{matched}{}} & \ifcsvstrcmp{0.5871}{matched}{\matchedcellslash{Qwen}}{\replacementpair{}{0.5871}{0.058}} & \ifcsvstrcmp{0.4288}{matched}{\matchedcellslash{OpenLLaMA}}{\replacementpair{}{0.4288}{0.006}} & \ifcsvstrcmp{0.4766}{matched}{\matchedcellslash{OPT}}{\replacementpair{}{0.4766}{0.004}} & \ifcsvstrcmp{0.7159}{matched}{\matchedcellslash{GPT-Neo}}{\replacementpair{}{0.7159}{0.104}} & \ifcsvstrcmp{0.8081}{matched}{\matchedcellslash{TinyLLaMA}}{\replacementpair{}{0.8081}{0.086}}\\\addlinespace[0pt]
\ifcsvstrcmp{0}{1}{\addlinespace[3pt]}{} & Qwen & \matchedpair{0.9887}{0.914} & \ifcsvstrcmp{0.4894}{matched}{\matchedcellslash{Pythia}}{\replacementpair{}{0.4894}{0.018}} & \ifcsvstrcmp{matched}{matched}{\matchedcellslash{Qwen}}{\replacementpair{}{matched}{}} & \ifcsvstrcmp{0.6036}{matched}{\matchedcellslash{OpenLLaMA}}{\replacementpair{}{0.6036}{0.012}} & \ifcsvstrcmp{0.3878}{matched}{\matchedcellslash{OPT}}{\replacementpair{}{0.3878}{0.006}} & \ifcsvstrcmp{0.7836}{matched}{\matchedcellslash{GPT-Neo}}{\replacementpair{}{0.7836}{0.144}} & \ifcsvstrcmp{0.6972}{matched}{\matchedcellslash{TinyLLaMA}}{\replacementpair{}{0.6972}{0.028}}\\\addlinespace[0pt]
\ifcsvstrcmp{0}{1}{\addlinespace[3pt]}{} & OpenLLaMA & \matchedpair{0.9991}{0.982} & \ifcsvstrcmp{0.4727}{matched}{\matchedcellslash{Pythia}}{\replacementpair{}{0.4727}{0.056}} & \ifcsvstrcmp{0.6503}{matched}{\matchedcellslash{Qwen}}{\replacementpair{}{0.6503}{0.120}} & \ifcsvstrcmp{matched}{matched}{\matchedcellslash{OpenLLaMA}}{\replacementpair{}{matched}{}} & \ifcsvstrcmp{0.5981}{matched}{\matchedcellslash{OPT}}{\replacementpair{}{0.5981}{0.040}} & \ifcsvstrcmp{0.7228}{matched}{\matchedcellslash{GPT-Neo}}{\replacementpair{}{0.7228}{0.106}} & \ifcsvstrcmp{0.7492}{matched}{\matchedcellslash{TinyLLaMA}}{\replacementpair{}{0.7492}{0.090}}\\\addlinespace[0pt]
\ifcsvstrcmp{1}{1}{\addlinespace[3pt]}{}LFQA & Pythia & \matchedpair{0.9987}{0.986} & \ifcsvstrcmp{matched}{matched}{\matchedcellslash{Pythia}}{\replacementpair{}{matched}{}} & \ifcsvstrcmp{0.4372}{matched}{\matchedcellslash{Qwen}}{\replacementpair{}{0.4372}{0.014}} & \ifcsvstrcmp{0.4291}{matched}{\matchedcellslash{OpenLLaMA}}{\replacementpair{}{0.4291}{0.006}} & \ifcsvstrcmp{0.5706}{matched}{\matchedcellslash{OPT}}{\replacementpair{}{0.5706}{0.024}} & \ifcsvstrcmp{0.5803}{matched}{\matchedcellslash{GPT-Neo}}{\replacementpair{}{0.5803}{0.030}} & \ifcsvstrcmp{0.6609}{matched}{\matchedcellslash{TinyLLaMA}}{\replacementpair{}{0.6609}{0.038}}\\\addlinespace[0pt]
\ifcsvstrcmp{0}{1}{\addlinespace[3pt]}{} & Qwen & \matchedpair{0.9863}{0.906} & \ifcsvstrcmp{0.4568}{matched}{\matchedcellslash{Pythia}}{\replacementpair{}{0.4568}{0.006}} & \ifcsvstrcmp{matched}{matched}{\matchedcellslash{Qwen}}{\replacementpair{}{matched}{}} & \ifcsvstrcmp{0.4320}{matched}{\matchedcellslash{OpenLLaMA}}{\replacementpair{}{0.4320}{0.004}} & \ifcsvstrcmp{0.3886}{matched}{\matchedcellslash{OPT}}{\replacementpair{}{0.3886}{0.004}} & \ifcsvstrcmp{0.5924}{matched}{\matchedcellslash{GPT-Neo}}{\replacementpair{}{0.5924}{0.036}} & \ifcsvstrcmp{0.5196}{matched}{\matchedcellslash{TinyLLaMA}}{\replacementpair{}{0.5196}{0.006}}\\\addlinespace[0pt]
\ifcsvstrcmp{0}{1}{\addlinespace[3pt]}{} & OpenLLaMA & \matchedpair{0.9995}{0.988} & \ifcsvstrcmp{0.3126}{matched}{\matchedcellslash{Pythia}}{\replacementpair{}{0.3126}{0.006}} & \ifcsvstrcmp{0.4773}{matched}{\matchedcellslash{Qwen}}{\replacementpair{}{0.4773}{0.032}} & \ifcsvstrcmp{matched}{matched}{\matchedcellslash{OpenLLaMA}}{\replacementpair{}{matched}{}} & \ifcsvstrcmp{0.5401}{matched}{\matchedcellslash{OPT}}{\replacementpair{}{0.5401}{0.016}} & \ifcsvstrcmp{0.6114}{matched}{\matchedcellslash{GPT-Neo}}{\replacementpair{}{0.6114}{0.074}} & \ifcsvstrcmp{0.5138}{matched}{\matchedcellslash{TinyLLaMA}}{\replacementpair{}{0.5138}{0.028}}\\\addlinespace[0pt]
\bottomrule
\end{tabular}
}
\endgroup
\par\smallskip\begin{minipage}{\linewidth}\footnotesize
The matched reference is shown in the Native column; a diagonal stroke marks that same tokenizer, not a missing result.
\end{minipage}
\end{table}

\begin{table}[H]
\centering\small
\caption{EXP: matched versus replacement tokenizers. Each cell shows AUROC (top) and TPR (bottom). Cells with TPR below $0.2$ are shaded pale yellow.}
\label{tab:e02-exp}
\begingroup
\setlength{\tabcolsep}{3.5pt}
\renewcommand{\arraystretch}{1.0}
\maxsizebox{\linewidth}{!}{%
\begin{tabular}{@{}llc@{\hspace{8pt}}cccccc@{}}
\toprule
 & & \multicolumn{1}{c}{Matched} & \multicolumn{6}{c}{Replacement detector tokenizer} \\
\cmidrule(lr){3-3}\cmidrule(l){4-9}
Domain & Generator & Native & Pythia & Qwen & OpenLLaMA & OPT & GPT-Neo & TinyLLaMA \\
\midrule
\ifcsvstrcmp{0}{1}{\addlinespace[3pt]}{}C4 & Pythia & \ifcsvstrcmp{0.9982}{--}{--}{\matchedpair{0.9982}{0.996}} & \ifcsvstrcmp{matched}{matched}{\matchedcellslash{Pythia}}{\ifcsvstrcmp{matched}{--}{--}{\replacementpair{}{matched}{}}} & \ifcsvstrcmp{0.5282}{matched}{\matchedcellslash{Qwen}}{\ifcsvstrcmp{0.5282}{--}{--}{\replacementpair{}{0.5282}{0.006}}} & \ifcsvstrcmp{0.5313}{matched}{\matchedcellslash{OpenLLaMA}}{\ifcsvstrcmp{0.5313}{--}{--}{\replacementpair{}{0.5313}{0.026}}} & \ifcsvstrcmp{0.5131}{matched}{\matchedcellslash{OPT}}{\ifcsvstrcmp{0.5131}{--}{--}{\replacementpair{}{0.5131}{0.008}}} & \ifcsvstrcmp{0.5280}{matched}{\matchedcellslash{GPT-Neo}}{\ifcsvstrcmp{0.5280}{--}{--}{\replacementpair{}{0.5280}{0.006}}} & \ifcsvstrcmp{0.5180}{matched}{\matchedcellslash{TinyLLaMA}}{\ifcsvstrcmp{0.5180}{--}{--}{\replacementpair{}{0.5180}{0.010}}}\\\addlinespace[0pt]
\ifcsvstrcmp{0}{1}{\addlinespace[3pt]}{} & Qwen & \ifcsvstrcmp{1.0000}{--}{--}{\matchedpair{1.0000}{1.000}} & \ifcsvstrcmp{0.4874}{matched}{\matchedcellslash{Pythia}}{\ifcsvstrcmp{0.4874}{--}{--}{\replacementpair{}{0.4874}{0.004}}} & \ifcsvstrcmp{matched}{matched}{\matchedcellslash{Qwen}}{\ifcsvstrcmp{matched}{--}{--}{\replacementpair{}{matched}{}}} & \ifcsvstrcmp{0.5620}{matched}{\matchedcellslash{OpenLLaMA}}{\ifcsvstrcmp{0.5620}{--}{--}{\replacementpair{}{0.5620}{0.054}}} & \ifcsvstrcmp{0.4876}{matched}{\matchedcellslash{OPT}}{\ifcsvstrcmp{0.4876}{--}{--}{\replacementpair{}{0.4876}{0.014}}} & \ifcsvstrcmp{0.4926}{matched}{\matchedcellslash{GPT-Neo}}{\ifcsvstrcmp{0.4926}{--}{--}{\replacementpair{}{0.4926}{0.010}}} & \ifcsvstrcmp{0.5202}{matched}{\matchedcellslash{TinyLLaMA}}{\ifcsvstrcmp{0.5202}{--}{--}{\replacementpair{}{0.5202}{0.020}}}\\\addlinespace[0pt]
\ifcsvstrcmp{0}{1}{\addlinespace[3pt]}{} & OpenLLaMA & \ifcsvstrcmp{0.9990}{--}{--}{\matchedpair{0.9990}{0.996}} & \ifcsvstrcmp{0.4713}{matched}{\matchedcellslash{Pythia}}{\ifcsvstrcmp{0.4713}{--}{--}{\replacementpair{}{0.4713}{0.014}}} & \ifcsvstrcmp{0.4928}{matched}{\matchedcellslash{Qwen}}{\ifcsvstrcmp{0.4928}{--}{--}{\replacementpair{}{0.4928}{0.014}}} & \ifcsvstrcmp{matched}{matched}{\matchedcellslash{OpenLLaMA}}{\ifcsvstrcmp{matched}{--}{--}{\replacementpair{}{matched}{}}} & \ifcsvstrcmp{0.4923}{matched}{\matchedcellslash{OPT}}{\ifcsvstrcmp{0.4923}{--}{--}{\replacementpair{}{0.4923}{0.014}}} & \ifcsvstrcmp{0.5125}{matched}{\matchedcellslash{GPT-Neo}}{\ifcsvstrcmp{0.5125}{--}{--}{\replacementpair{}{0.5125}{0.016}}} & \ifcsvstrcmp{0.5157}{matched}{\matchedcellslash{TinyLLaMA}}{\ifcsvstrcmp{0.5157}{--}{--}{\replacementpair{}{0.5157}{0.016}}}\\\addlinespace[0pt]
\ifcsvstrcmp{1}{1}{\addlinespace[3pt]}{}Wiki & Pythia & \ifcsvstrcmp{1.0000}{--}{--}{\matchedpair{1.0000}{1.000}} & \ifcsvstrcmp{matched}{matched}{\matchedcellslash{Pythia}}{\ifcsvstrcmp{matched}{--}{--}{\replacementpair{}{matched}{}}} & \ifcsvstrcmp{0.5452}{matched}{\matchedcellslash{Qwen}}{\ifcsvstrcmp{0.5452}{--}{--}{\replacementpair{}{0.5452}{0.006}}} & \ifcsvstrcmp{0.5075}{matched}{\matchedcellslash{OpenLLaMA}}{\ifcsvstrcmp{0.5075}{--}{--}{\replacementpair{}{0.5075}{0.018}}} & \ifcsvstrcmp{0.4644}{matched}{\matchedcellslash{OPT}}{\ifcsvstrcmp{0.4644}{--}{--}{\replacementpair{}{0.4644}{0.004}}} & \ifcsvstrcmp{0.4763}{matched}{\matchedcellslash{GPT-Neo}}{\ifcsvstrcmp{0.4763}{--}{--}{\replacementpair{}{0.4763}{0.010}}} & \ifcsvstrcmp{0.4979}{matched}{\matchedcellslash{TinyLLaMA}}{\ifcsvstrcmp{0.4979}{--}{--}{\replacementpair{}{0.4979}{0.016}}}\\\addlinespace[0pt]
\ifcsvstrcmp{0}{1}{\addlinespace[3pt]}{} & Qwen & \ifcsvstrcmp{1.0000}{--}{--}{\matchedpair{1.0000}{1.000}} & \ifcsvstrcmp{0.5067}{matched}{\matchedcellslash{Pythia}}{\ifcsvstrcmp{0.5067}{--}{--}{\replacementpair{}{0.5067}{0.020}}} & \ifcsvstrcmp{matched}{matched}{\matchedcellslash{Qwen}}{\ifcsvstrcmp{matched}{--}{--}{\replacementpair{}{matched}{}}} & \ifcsvstrcmp{0.4888}{matched}{\matchedcellslash{OpenLLaMA}}{\ifcsvstrcmp{0.4888}{--}{--}{\replacementpair{}{0.4888}{0.018}}} & \ifcsvstrcmp{0.4471}{matched}{\matchedcellslash{OPT}}{\ifcsvstrcmp{0.4471}{--}{--}{\replacementpair{}{0.4471}{0.004}}} & \ifcsvstrcmp{0.5104}{matched}{\matchedcellslash{GPT-Neo}}{\ifcsvstrcmp{0.5104}{--}{--}{\replacementpair{}{0.5104}{0.018}}} & \ifcsvstrcmp{0.4714}{matched}{\matchedcellslash{TinyLLaMA}}{\ifcsvstrcmp{0.4714}{--}{--}{\replacementpair{}{0.4714}{0.008}}}\\\addlinespace[0pt]
\ifcsvstrcmp{0}{1}{\addlinespace[3pt]}{} & OpenLLaMA & \ifcsvstrcmp{1.0000}{--}{--}{\matchedpair{1.0000}{1.000}} & \ifcsvstrcmp{0.5532}{matched}{\matchedcellslash{Pythia}}{\ifcsvstrcmp{0.5532}{--}{--}{\replacementpair{}{0.5532}{0.012}}} & \ifcsvstrcmp{0.5589}{matched}{\matchedcellslash{Qwen}}{\ifcsvstrcmp{0.5589}{--}{--}{\replacementpair{}{0.5589}{0.010}}} & \ifcsvstrcmp{matched}{matched}{\matchedcellslash{OpenLLaMA}}{\ifcsvstrcmp{matched}{--}{--}{\replacementpair{}{matched}{}}} & \ifcsvstrcmp{0.4982}{matched}{\matchedcellslash{OPT}}{\ifcsvstrcmp{0.4982}{--}{--}{\replacementpair{}{0.4982}{0.012}}} & \ifcsvstrcmp{0.4800}{matched}{\matchedcellslash{GPT-Neo}}{\ifcsvstrcmp{0.4800}{--}{--}{\replacementpair{}{0.4800}{0.008}}} & \ifcsvstrcmp{0.4611}{matched}{\matchedcellslash{TinyLLaMA}}{\ifcsvstrcmp{0.4611}{--}{--}{\replacementpair{}{0.4611}{0.016}}}\\\addlinespace[0pt]
\ifcsvstrcmp{1}{1}{\addlinespace[3pt]}{}LFQA & Pythia & \ifcsvstrcmp{0.9990}{--}{--}{\matchedpair{0.9990}{1.000}} & \ifcsvstrcmp{matched}{matched}{\matchedcellslash{Pythia}}{\ifcsvstrcmp{matched}{--}{--}{\replacementpair{}{matched}{}}} & \ifcsvstrcmp{0.5139}{matched}{\matchedcellslash{Qwen}}{\ifcsvstrcmp{0.5139}{--}{--}{\replacementpair{}{0.5139}{0.018}}} & \ifcsvstrcmp{0.5202}{matched}{\matchedcellslash{OpenLLaMA}}{\ifcsvstrcmp{0.5202}{--}{--}{\replacementpair{}{0.5202}{0.016}}} & \ifcsvstrcmp{0.5184}{matched}{\matchedcellslash{OPT}}{\ifcsvstrcmp{0.5184}{--}{--}{\replacementpair{}{0.5184}{0.004}}} & \ifcsvstrcmp{0.4909}{matched}{\matchedcellslash{GPT-Neo}}{\ifcsvstrcmp{0.4909}{--}{--}{\replacementpair{}{0.4909}{0.010}}} & \ifcsvstrcmp{0.5187}{matched}{\matchedcellslash{TinyLLaMA}}{\ifcsvstrcmp{0.5187}{--}{--}{\replacementpair{}{0.5187}{0.020}}}\\\addlinespace[0pt]
\ifcsvstrcmp{0}{1}{\addlinespace[3pt]}{} & Qwen & \ifcsvstrcmp{0.9998}{--}{--}{\matchedpair{0.9998}{0.998}} & \ifcsvstrcmp{0.5159}{matched}{\matchedcellslash{Pythia}}{\ifcsvstrcmp{0.5159}{--}{--}{\replacementpair{}{0.5159}{0.010}}} & \ifcsvstrcmp{matched}{matched}{\matchedcellslash{Qwen}}{\ifcsvstrcmp{matched}{--}{--}{\replacementpair{}{matched}{}}} & \ifcsvstrcmp{0.4870}{matched}{\matchedcellslash{OpenLLaMA}}{\ifcsvstrcmp{0.4870}{--}{--}{\replacementpair{}{0.4870}{0.012}}} & \ifcsvstrcmp{0.4885}{matched}{\matchedcellslash{OPT}}{\ifcsvstrcmp{0.4885}{--}{--}{\replacementpair{}{0.4885}{0.004}}} & \ifcsvstrcmp{0.4711}{matched}{\matchedcellslash{GPT-Neo}}{\ifcsvstrcmp{0.4711}{--}{--}{\replacementpair{}{0.4711}{0.006}}} & \ifcsvstrcmp{0.4896}{matched}{\matchedcellslash{TinyLLaMA}}{\ifcsvstrcmp{0.4896}{--}{--}{\replacementpair{}{0.4896}{0.016}}}\\\addlinespace[0pt]
\ifcsvstrcmp{0}{1}{\addlinespace[3pt]}{} & OpenLLaMA & \ifcsvstrcmp{0.9995}{--}{--}{\matchedpair{0.9995}{1.000}} & \ifcsvstrcmp{0.5016}{matched}{\matchedcellslash{Pythia}}{\ifcsvstrcmp{0.5016}{--}{--}{\replacementpair{}{0.5016}{0.018}}} & \ifcsvstrcmp{0.4888}{matched}{\matchedcellslash{Qwen}}{\ifcsvstrcmp{0.4888}{--}{--}{\replacementpair{}{0.4888}{0.026}}} & \ifcsvstrcmp{matched}{matched}{\matchedcellslash{OpenLLaMA}}{\ifcsvstrcmp{matched}{--}{--}{\replacementpair{}{matched}{}}} & \ifcsvstrcmp{0.5129}{matched}{\matchedcellslash{OPT}}{\ifcsvstrcmp{0.5129}{--}{--}{\replacementpair{}{0.5129}{0.010}}} & \ifcsvstrcmp{0.4666}{matched}{\matchedcellslash{GPT-Neo}}{\ifcsvstrcmp{0.4666}{--}{--}{\replacementpair{}{0.4666}{0.004}}} & \ifcsvstrcmp{0.4935}{matched}{\matchedcellslash{TinyLLaMA}}{\ifcsvstrcmp{0.4935}{--}{--}{\replacementpair{}{0.4935}{0.022}}}\\\addlinespace[0pt]
\bottomrule
\end{tabular}
}
\endgroup
\par\smallskip\begin{minipage}{\linewidth}\footnotesize
The matched reference is shown in the Native column; a diagonal stroke marks that same tokenizer, not a missing result.
\end{minipage}
\end{table}

\begin{table}[H]
\centering\small
\caption{ITS: matched versus replacement tokenizers. Each cell shows AUROC (top) and TPR (bottom). Cells with TPR below $0.2$ are shaded pale yellow.}
\label{tab:e02-its}
\begingroup
\setlength{\tabcolsep}{3.5pt}
\renewcommand{\arraystretch}{1.0}
\maxsizebox{\linewidth}{!}{%
\begin{tabular}{@{}llc@{\hspace{8pt}}cccccc@{}}
\toprule
 & & \multicolumn{1}{c}{Matched} & \multicolumn{6}{c}{Replacement detector tokenizer} \\
\cmidrule(lr){3-3}\cmidrule(l){4-9}
Domain & Generator & Native & Pythia & Qwen & OpenLLaMA & OPT & GPT-Neo & TinyLLaMA \\
\midrule
\ifcsvstrcmp{0}{1}{\addlinespace[3pt]}{}C4 & Pythia & \ifcsvstrcmp{0.9911}{--}{--}{\matchedpair{0.9911}{0.980}} & \ifcsvstrcmp{matched}{matched}{\matchedcellslash{Pythia}}{\ifcsvstrcmp{matched}{--}{--}{\replacementpair{}{matched}{}}} & \ifcsvstrcmp{0.4878}{matched}{\matchedcellslash{Qwen}}{\ifcsvstrcmp{0.4878}{--}{--}{\replacementpair{}{0.4878}{0.006}}} & \ifcsvstrcmp{0.4877}{matched}{\matchedcellslash{OpenLLaMA}}{\ifcsvstrcmp{0.4877}{--}{--}{\replacementpair{}{0.4877}{0.006}}} & \ifcsvstrcmp{0.4936}{matched}{\matchedcellslash{OPT}}{\ifcsvstrcmp{0.4936}{--}{--}{\replacementpair{}{0.4936}{0.002}}} & \ifcsvstrcmp{0.4394}{matched}{\matchedcellslash{GPT-Neo}}{\ifcsvstrcmp{0.4394}{--}{--}{\replacementpair{}{0.4394}{0.010}}} & \ifcsvstrcmp{0.5230}{matched}{\matchedcellslash{TinyLLaMA}}{\ifcsvstrcmp{0.5230}{--}{--}{\replacementpair{}{0.5230}{0.004}}}\\\addlinespace[0pt]
\ifcsvstrcmp{0}{1}{\addlinespace[3pt]}{} & Qwen & \ifcsvstrcmp{0.9994}{--}{--}{\matchedpair{0.9994}{0.992}} & \ifcsvstrcmp{0.4831}{matched}{\matchedcellslash{Pythia}}{\ifcsvstrcmp{0.4831}{--}{--}{\replacementpair{}{0.4831}{0.010}}} & \ifcsvstrcmp{matched}{matched}{\matchedcellslash{Qwen}}{\ifcsvstrcmp{matched}{--}{--}{\replacementpair{}{matched}{}}} & \ifcsvstrcmp{0.4997}{matched}{\matchedcellslash{OpenLLaMA}}{\ifcsvstrcmp{0.4997}{--}{--}{\replacementpair{}{0.4997}{0.016}}} & \ifcsvstrcmp{0.4874}{matched}{\matchedcellslash{OPT}}{\ifcsvstrcmp{0.4874}{--}{--}{\replacementpair{}{0.4874}{0.004}}} & \ifcsvstrcmp{0.4852}{matched}{\matchedcellslash{GPT-Neo}}{\ifcsvstrcmp{0.4852}{--}{--}{\replacementpair{}{0.4852}{0.004}}} & \ifcsvstrcmp{0.4812}{matched}{\matchedcellslash{TinyLLaMA}}{\ifcsvstrcmp{0.4812}{--}{--}{\replacementpair{}{0.4812}{0.004}}}\\\addlinespace[0pt]
\ifcsvstrcmp{0}{1}{\addlinespace[3pt]}{} & OpenLLaMA & \ifcsvstrcmp{0.9944}{--}{--}{\matchedpair{0.9944}{0.980}} & \ifcsvstrcmp{0.4988}{matched}{\matchedcellslash{Pythia}}{\ifcsvstrcmp{0.4988}{--}{--}{\replacementpair{}{0.4988}{0.016}}} & \ifcsvstrcmp{0.5271}{matched}{\matchedcellslash{Qwen}}{\ifcsvstrcmp{0.5271}{--}{--}{\replacementpair{}{0.5271}{0.008}}} & \ifcsvstrcmp{matched}{matched}{\matchedcellslash{OpenLLaMA}}{\ifcsvstrcmp{matched}{--}{--}{\replacementpair{}{matched}{}}} & \ifcsvstrcmp{0.4894}{matched}{\matchedcellslash{OPT}}{\ifcsvstrcmp{0.4894}{--}{--}{\replacementpair{}{0.4894}{0.008}}} & \ifcsvstrcmp{0.4400}{matched}{\matchedcellslash{GPT-Neo}}{\ifcsvstrcmp{0.4400}{--}{--}{\replacementpair{}{0.4400}{0.008}}} & \ifcsvstrcmp{0.5652}{matched}{\matchedcellslash{TinyLLaMA}}{\ifcsvstrcmp{0.5652}{--}{--}{\replacementpair{}{0.5652}{0.016}}}\\\addlinespace[0pt]
\ifcsvstrcmp{1}{1}{\addlinespace[3pt]}{}Wiki & Pythia & \ifcsvstrcmp{0.9994}{--}{--}{\matchedpair{0.9994}{0.994}} & \ifcsvstrcmp{matched}{matched}{\matchedcellslash{Pythia}}{\ifcsvstrcmp{matched}{--}{--}{\replacementpair{}{matched}{}}} & \ifcsvstrcmp{0.4388}{matched}{\matchedcellslash{Qwen}}{\ifcsvstrcmp{0.4388}{--}{--}{\replacementpair{}{0.4388}{0.008}}} & \ifcsvstrcmp{0.3487}{matched}{\matchedcellslash{OpenLLaMA}}{\ifcsvstrcmp{0.3487}{--}{--}{\replacementpair{}{0.3487}{0.002}}} & \ifcsvstrcmp{0.4777}{matched}{\matchedcellslash{OPT}}{\ifcsvstrcmp{0.4777}{--}{--}{\replacementpair{}{0.4777}{0.004}}} & \ifcsvstrcmp{0.4854}{matched}{\matchedcellslash{GPT-Neo}}{\ifcsvstrcmp{0.4854}{--}{--}{\replacementpair{}{0.4854}{0.012}}} & \ifcsvstrcmp{0.2893}{matched}{\matchedcellslash{TinyLLaMA}}{\ifcsvstrcmp{0.2893}{--}{--}{\replacementpair{}{0.2893}{0.000}}}\\\addlinespace[0pt]
\ifcsvstrcmp{0}{1}{\addlinespace[3pt]}{} & Qwen & \ifcsvstrcmp{0.9877}{--}{--}{\matchedpair{0.9877}{0.928}} & \ifcsvstrcmp{0.5881}{matched}{\matchedcellslash{Pythia}}{\ifcsvstrcmp{0.5881}{--}{--}{\replacementpair{}{0.5881}{0.012}}} & \ifcsvstrcmp{matched}{matched}{\matchedcellslash{Qwen}}{\ifcsvstrcmp{matched}{--}{--}{\replacementpair{}{matched}{}}} & \ifcsvstrcmp{0.3527}{matched}{\matchedcellslash{OpenLLaMA}}{\ifcsvstrcmp{0.3527}{--}{--}{\replacementpair{}{0.3527}{0.002}}} & \ifcsvstrcmp{0.4898}{matched}{\matchedcellslash{OPT}}{\ifcsvstrcmp{0.4898}{--}{--}{\replacementpair{}{0.4898}{0.006}}} & \ifcsvstrcmp{0.3925}{matched}{\matchedcellslash{GPT-Neo}}{\ifcsvstrcmp{0.3925}{--}{--}{\replacementpair{}{0.3925}{0.004}}} & \ifcsvstrcmp{0.3406}{matched}{\matchedcellslash{TinyLLaMA}}{\ifcsvstrcmp{0.3406}{--}{--}{\replacementpair{}{0.3406}{0.000}}}\\\addlinespace[0pt]
\ifcsvstrcmp{0}{1}{\addlinespace[3pt]}{} & OpenLLaMA & \ifcsvstrcmp{0.9953}{--}{--}{\matchedpair{0.9953}{0.982}} & \ifcsvstrcmp{0.5820}{matched}{\matchedcellslash{Pythia}}{\ifcsvstrcmp{0.5820}{--}{--}{\replacementpair{}{0.5820}{0.018}}} & \ifcsvstrcmp{0.4687}{matched}{\matchedcellslash{Qwen}}{\ifcsvstrcmp{0.4687}{--}{--}{\replacementpair{}{0.4687}{0.012}}} & \ifcsvstrcmp{matched}{matched}{\matchedcellslash{OpenLLaMA}}{\ifcsvstrcmp{matched}{--}{--}{\replacementpair{}{matched}{}}} & \ifcsvstrcmp{0.5016}{matched}{\matchedcellslash{OPT}}{\ifcsvstrcmp{0.5016}{--}{--}{\replacementpair{}{0.5016}{0.016}}} & \ifcsvstrcmp{0.4581}{matched}{\matchedcellslash{GPT-Neo}}{\ifcsvstrcmp{0.4581}{--}{--}{\replacementpair{}{0.4581}{0.016}}} & \ifcsvstrcmp{0.3501}{matched}{\matchedcellslash{TinyLLaMA}}{\ifcsvstrcmp{0.3501}{--}{--}{\replacementpair{}{0.3501}{0.006}}}\\\addlinespace[0pt]
\ifcsvstrcmp{1}{1}{\addlinespace[3pt]}{}LFQA & Pythia & \ifcsvstrcmp{0.9987}{--}{--}{\matchedpair{0.9987}{0.994}} & \ifcsvstrcmp{matched}{matched}{\matchedcellslash{Pythia}}{\ifcsvstrcmp{matched}{--}{--}{\replacementpair{}{matched}{}}} & \ifcsvstrcmp{0.5069}{matched}{\matchedcellslash{Qwen}}{\ifcsvstrcmp{0.5069}{--}{--}{\replacementpair{}{0.5069}{0.006}}} & \ifcsvstrcmp{0.5674}{matched}{\matchedcellslash{OpenLLaMA}}{\ifcsvstrcmp{0.5674}{--}{--}{\replacementpair{}{0.5674}{0.014}}} & \ifcsvstrcmp{0.5384}{matched}{\matchedcellslash{OPT}}{\ifcsvstrcmp{0.5384}{--}{--}{\replacementpair{}{0.5384}{0.012}}} & \ifcsvstrcmp{0.4027}{matched}{\matchedcellslash{GPT-Neo}}{\ifcsvstrcmp{0.4027}{--}{--}{\replacementpair{}{0.4027}{0.000}}} & \ifcsvstrcmp{0.5226}{matched}{\matchedcellslash{TinyLLaMA}}{\ifcsvstrcmp{0.5226}{--}{--}{\replacementpair{}{0.5226}{0.010}}}\\\addlinespace[0pt]
\ifcsvstrcmp{0}{1}{\addlinespace[3pt]}{} & Qwen & \ifcsvstrcmp{0.9949}{--}{--}{\matchedpair{0.9949}{0.980}} & \ifcsvstrcmp{0.5132}{matched}{\matchedcellslash{Pythia}}{\ifcsvstrcmp{0.5132}{--}{--}{\replacementpair{}{0.5132}{0.014}}} & \ifcsvstrcmp{matched}{matched}{\matchedcellslash{Qwen}}{\ifcsvstrcmp{matched}{--}{--}{\replacementpair{}{matched}{}}} & \ifcsvstrcmp{0.5784}{matched}{\matchedcellslash{OpenLLaMA}}{\ifcsvstrcmp{0.5784}{--}{--}{\replacementpair{}{0.5784}{0.010}}} & \ifcsvstrcmp{0.5671}{matched}{\matchedcellslash{OPT}}{\ifcsvstrcmp{0.5671}{--}{--}{\replacementpair{}{0.5671}{0.010}}} & \ifcsvstrcmp{0.4335}{matched}{\matchedcellslash{GPT-Neo}}{\ifcsvstrcmp{0.4335}{--}{--}{\replacementpair{}{0.4335}{0.004}}} & \ifcsvstrcmp{0.5153}{matched}{\matchedcellslash{TinyLLaMA}}{\ifcsvstrcmp{0.5153}{--}{--}{\replacementpair{}{0.5153}{0.010}}}\\\addlinespace[0pt]
\ifcsvstrcmp{0}{1}{\addlinespace[3pt]}{} & OpenLLaMA & \ifcsvstrcmp{0.9982}{--}{--}{\matchedpair{0.9982}{0.990}} & \ifcsvstrcmp{0.5319}{matched}{\matchedcellslash{Pythia}}{\ifcsvstrcmp{0.5319}{--}{--}{\replacementpair{}{0.5319}{0.010}}} & \ifcsvstrcmp{0.5322}{matched}{\matchedcellslash{Qwen}}{\ifcsvstrcmp{0.5322}{--}{--}{\replacementpair{}{0.5322}{0.012}}} & \ifcsvstrcmp{matched}{matched}{\matchedcellslash{OpenLLaMA}}{\ifcsvstrcmp{matched}{--}{--}{\replacementpair{}{matched}{}}} & \ifcsvstrcmp{0.5503}{matched}{\matchedcellslash{OPT}}{\ifcsvstrcmp{0.5503}{--}{--}{\replacementpair{}{0.5503}{0.016}}} & \ifcsvstrcmp{0.3711}{matched}{\matchedcellslash{GPT-Neo}}{\ifcsvstrcmp{0.3711}{--}{--}{\replacementpair{}{0.3711}{0.006}}} & \ifcsvstrcmp{0.5538}{matched}{\matchedcellslash{TinyLLaMA}}{\ifcsvstrcmp{0.5538}{--}{--}{\replacementpair{}{0.5538}{0.016}}}\\\addlinespace[0pt]
\bottomrule
\end{tabular}
}
\endgroup
\par\smallskip\begin{minipage}{\linewidth}\footnotesize
The matched reference is shown in the Native column; a diagonal stroke marks that same tokenizer, not a missing result.
\end{minipage}
\end{table}

\endgroup
Tables~\ref{tab:e02-exp} and~\ref{tab:e02-its} give the same matrices for EXP and
ITS. Matched detection reaches AUROC of at least $0.98$ in every group, whereas the
replacement cells average AUROC $0.502$ for EXP and $0.481$ for ITS, with TPR of
$0.014$ and $0.009$.

\subsection{Encoder and generator replacement for semantic methods}
\label{sec:encoder-replacement}
\label{sec:generator-replacement}
Semantic watermarks score text in the embedding space of a specified encoder,
and their detectors store artifacts derived from that space, such as cluster
centers or a calibrated similarity interval. This experiment holds the generated
text fixed and replaces only the encoder used at detection, leaving every stored
artifact unchanged. Each method is evaluated with the five encoders common to all
methods; the matched embedding model is the one used at generation. A replacement whose
output dimension differs from that of a stored artifact defines no detector and
is reported as undefined rather than as a score of zero.

Table~\ref{tab:e5-author-c4} reports the study on C4 with each method's author
generation model. The matched embedding model detects the watermark reliably. After
replacement, TPR falls in every scoreable configuration. Ranking performance varies across replacements: most AUROC values are near or below chance, while SimMark
with instructor-base in place of
instructor-large retains partial discrimination. Both encoders belong to the instructor model family
\citep{su2023instructor,ni2022gtr}, suggesting that the replacement preserves
some of the embedding-space structure used by the detector.
\begin{table}[t]
\centering\small
\caption{Detection-encoder replacement on C4 with each method's author generation model. The matched embedding model is the one used at generation. TPR cells below $0.2$ are shaded pale yellow.}
\label{tab:e5-author-c4}
\begin{tabular}{@{}lllrr@{}}
\toprule
Method & Detection encoder & Relation & AUROC & TPR \\
\midrule
\ifcsvstrcmp{0}{1}{\midrule}{}SimMark & instructor-large & matched & 0.9889 & \lowtprcell{0.9178}\\
\ifcsvstrcmp{0}{1}{\midrule}{} & semstamp-c4-sbert & replacement & 0.3448 & \lowtprcell{0.0080}\\
\ifcsvstrcmp{0}{1}{\midrule}{} & instructor-base & replacement & 0.9051 & \lowtprcell{0.3006}\\
\ifcsvstrcmp{0}{1}{\midrule}{} & mpnet-base-v2 & replacement & 0.2982 & \lowtprcell{0.0240}\\
\ifcsvstrcmp{0}{1}{\midrule}{} & MiniLM-L6-v2 & replacement & 0.3780 & \lowtprcell{0.0341}\\
\ifcsvstrcmp{1}{1}{\midrule}{}SemStamp & semstamp-c4-sbert & matched & 0.9996 & \lowtprcell{0.9914}\\
\ifcsvstrcmp{0}{1}{\midrule}{} & instructor-large & replacement & 0.4718 & \lowtprcell{0.0172}\\
\ifcsvstrcmp{0}{1}{\midrule}{} & instructor-base & replacement & 0.4517 & \lowtprcell{0.0237}\\
\ifcsvstrcmp{0}{1}{\midrule}{} & mpnet-base-v2 & replacement & 0.6314 & \lowtprcell{0.1013}\\
\ifcsvstrcmp{0}{1}{\midrule}{} & MiniLM-L6-v2 & replacement & 0.4741 & \lowtprcell{0.0108}\\
\ifcsvstrcmp{1}{1}{\midrule}{}k-SemStamp & semstamp-c4-sbert & matched & 0.9938 & \lowtprcell{0.9660}\\
\ifcsvstrcmp{0}{1}{\midrule}{} & instructor-large & replacement & 0.4851 & \lowtprcell{0.0140}\\
\ifcsvstrcmp{0}{1}{\midrule}{} & instructor-base & replacement & 0.5131 & \lowtprcell{0.0180}\\
\ifcsvstrcmp{0}{1}{\midrule}{} & mpnet-base-v2 & replacement & 0.7316 & \lowtprcell{0.0780}\\
\ifcsvstrcmp{0}{1}{\midrule}{} & MiniLM-L6-v2 & replacement & -- & \lowtprcell{--}\\
\ifcsvstrcmp{1}{1}{\midrule}{}AliMark & mpnet-base-v2 & matched & 0.9999 & \lowtprcell{0.9980}\\
\ifcsvstrcmp{0}{1}{\midrule}{} & semstamp-c4-sbert & replacement & 0.8452 & \lowtprcell{0.1880}\\
\ifcsvstrcmp{0}{1}{\midrule}{} & instructor-large & replacement & 0.4281 & \lowtprcell{0.0040}\\
\ifcsvstrcmp{0}{1}{\midrule}{} & instructor-base & replacement & 0.6267 & \lowtprcell{0.0020}\\
\ifcsvstrcmp{0}{1}{\midrule}{} & MiniLM-L6-v2 & replacement & -- & \lowtprcell{--}\\
\bottomrule
\end{tabular}

\par\smallskip\begin{minipage}{.94\linewidth}\footnotesize
The k-SemStamp/MiniLM-L6-v2 and AliMark/MiniLM-L6-v2 detectors are structurally undefined (768-dimensional frozen centers or secret vectors versus a 384-dimensional replacement), not zero. Operating points use the same-pool empirical 1\% threshold.
\end{minipage}
\end{table}

Table~\ref{tab:pythia-c4-wiki} evaluates the semantic methods with Pythia-1.4B
as a general-purpose generator, keeping the authors' watermark parameters
unchanged. For SimMark, SemStamp, and k-SemStamp, this replaces the fine-tuned
OPT-1.3B released with SemStamp \citep{hou2024semstamp}. Matched detection
performance declines for some methods; SimMark's TPR on C4 decreases from
$0.92$ to $0.62$. The Wiki rows extend the matched-detection evaluation to a
second dataset with the same generator and watermark parameters.

On the Pythia-generated C4 texts, encoder replacement also reduces TPR in every scoreable configuration.
\begin{table}[t]
\centering\small
\caption{Semantic methods with Pythia-1.4B as the generation model and author
watermark parameters. On C4 the detection encoder is replaced; on Wiki only the
matched embedding model is evaluated. The matched embedding model is the one used at generation. TPR cells below $0.2$ are shaded pale yellow.}
\label{tab:pythia-c4-wiki}
\begin{tabular}{@{}llllrr@{}}
\toprule
Method & Data & Detection encoder & Relation & AUROC & TPR \\
\midrule
\ifcsvstrcmp{0}{1}{\midrule}{}SimMark & C4 & instructor-large & matched & 0.9733 & \lowtprcell{0.6213}\\
\ifcsvstrcmp{0}{1}{\midrule}{} & Wiki & instructor-large & matched & 0.9968 & \lowtprcell{0.9299}\\
\ifcsvstrcmp{0}{1}{\midrule}{} & C4 & semstamp-c4-sbert & replacement & 0.4701 & \lowtprcell{0.0021}\\
\ifcsvstrcmp{0}{1}{\midrule}{} & C4 & instructor-base & replacement & 0.8592 & \lowtprcell{0.1444}\\
\ifcsvstrcmp{0}{1}{\midrule}{} & C4 & mpnet-base-v2 & replacement & 0.2650 & \lowtprcell{0.0084}\\
\ifcsvstrcmp{0}{1}{\midrule}{} & C4 & MiniLM-L6-v2 & replacement & 0.2492 & \lowtprcell{0.0167}\\
\ifcsvstrcmp{1}{1}{\midrule}{}SemStamp & C4 & semstamp-c4-sbert & matched & 0.9946 & \lowtprcell{0.9505}\\
\ifcsvstrcmp{0}{1}{\midrule}{} & Wiki & semstamp-c4-sbert & matched & 0.9964 & \lowtprcell{0.9281}\\
\ifcsvstrcmp{0}{1}{\midrule}{} & C4 & instructor-large & replacement & 0.5116 & \lowtprcell{0.0113}\\
\ifcsvstrcmp{0}{1}{\midrule}{} & C4 & instructor-base & replacement & 0.5540 & \lowtprcell{0.0135}\\
\ifcsvstrcmp{0}{1}{\midrule}{} & C4 & mpnet-base-v2 & replacement & 0.6104 & \lowtprcell{0.0383}\\
\ifcsvstrcmp{0}{1}{\midrule}{} & C4 & MiniLM-L6-v2 & replacement & 0.5396 & \lowtprcell{0.0045}\\
\ifcsvstrcmp{1}{1}{\midrule}{}k-SemStamp & C4 & semstamp-c4-sbert & matched & 0.9872 & \lowtprcell{0.9130}\\
\ifcsvstrcmp{0}{1}{\midrule}{} & Wiki & semstamp-c4-sbert & matched & 0.9901 & \lowtprcell{0.8621}\\
\ifcsvstrcmp{0}{1}{\midrule}{} & C4 & instructor-large & replacement & 0.5706 & \lowtprcell{0.0081}\\
\ifcsvstrcmp{0}{1}{\midrule}{} & C4 & instructor-base & replacement & 0.6033 & \lowtprcell{0.0101}\\
\ifcsvstrcmp{0}{1}{\midrule}{} & C4 & mpnet-base-v2 & replacement & 0.6740 & \lowtprcell{0.0425}\\
\ifcsvstrcmp{0}{1}{\midrule}{} & C4 & MiniLM-L6-v2 & replacement & -- & \lowtprcell{--}\\
\ifcsvstrcmp{1}{1}{\midrule}{}AliMark & C4 & mpnet-base-v2 & matched & 0.9969 & \lowtprcell{0.9940}\\
\ifcsvstrcmp{0}{1}{\midrule}{} & Wiki & mpnet-base-v2 & matched & 1.0000 & \lowtprcell{1.0000}\\
\ifcsvstrcmp{0}{1}{\midrule}{} & C4 & semstamp-c4-sbert & replacement & 0.8273 & \lowtprcell{0.1860}\\
\ifcsvstrcmp{0}{1}{\midrule}{} & C4 & instructor-large & replacement & 0.4053 & \lowtprcell{0.0000}\\
\ifcsvstrcmp{0}{1}{\midrule}{} & C4 & instructor-base & replacement & 0.6514 & \lowtprcell{0.0000}\\
\ifcsvstrcmp{0}{1}{\midrule}{} & C4 & MiniLM-L6-v2 & replacement & -- & \lowtprcell{--}\\
\bottomrule
\end{tabular}

\par\smallskip\begin{minipage}{.94\linewidth}\footnotesize
A dash denotes a structurally undefined detector: for k-SemStamp and AliMark, MiniLM-L6-v2 produces 384-dimensional embeddings incompatible with the stored 768-dimensional centers or secret vectors.
\end{minipage}
\end{table}

\section{Statistical correction over candidate tokenizers}
\label{sec:search}
\sectionlead{A token-level detector needs the tokenizer that produced the text,
and Section~\ref{sec:core} shows that with a wrong tokenizer, detection falls
to near chance level. In
practice the text often arrives without that information: the model behind an
online service is undisclosed, or the text has been collected without source
metadata. A verifier may still know the watermarking method and have access to
a finite set of plausible tokenizers, such as those of the public models in
use. This section analyzes that setting: the method is known, the generating
tokenizer is unknown, and the verifier scores the text under each of $K$
candidate tokenizers, obtaining $s_j(x)$ for candidate $j$, and must decide
from these $K$ scores together.}

\subsection{Why search requires correction}
A search declares a positive result when any candidate test rejects. Its
rejection event is the union of the candidate rejection events, so applying a
single-test threshold to every candidate can raise the overall false-positive
rate above the nominal level. Calibration must therefore account for the
combined decision.

\subsection{Calibration of the combined decision}
If each candidate supplies a valid null $p$-value $p_j$, so that
$\Pr_0(p_j\leq t)\leq t$ for $0\leq t\leq1$, Bonferroni correction rejects
when $\min_j p_j\leq\alpha/K$. The union bound gives
\[
 \Pr_0\!\left(\min_j p_j\leq\alpha/K\right)
 \leq\sum_{j=1}^K\Pr_0(p_j\leq\alpha/K)\leq\alpha.
\]
This guarantee requires neither independence nor identical score distributions.
When valid candidate $p$-values are unavailable, one can instead use
$M(x)=\max_j s_j(x)$ and calibrate its threshold from the joint null score
distribution. A threshold $\tau$ satisfying $\Pr_0(M>\tau)\leq\alpha$ controls
the combined false-positive rate. Empirically, the verifier computes all
candidate scores on each calibration text, takes their maximum, and uses the
empirical quantile of these maxima as the threshold.

\subsection{Detection cost of candidate search}
If the generating tokenizer is included, Bonferroni tests it at level $\alpha/K$ rather than $\alpha$, and
maximum-score calibration sets a threshold on $M$ that cannot be lower than the
single-test threshold, since adding candidates cannot decrease the maximum on
any text. As $K$ increases, the per-candidate level becomes smaller, which can reduce detection power relative to direct detection with the known tokenizer.

If the generating tokenizer is absent, the correction still controls the
false-positive rate, but no candidate reads the text in the representation that carries the watermark, so the search has almost no detection power.

\subsection{Empirical false-positive inflation}
The size of the cost depends on how much null variation the candidate scores
share. If the candidates' rejection events at their calibrated levels were identical, the search would cost nothing; if they were independent, the
search would pay the full multiplicity price, with family-wise rate
$1-(1-\alpha/K)^K$ for candidate-wise level $\alpha/K$, which differs from
$\alpha$ only at second order in $\alpha$.

Table~\ref{tab:search-main} in Section~\ref{sec:exp-replacement} reports the false-positive rate of an uncorrected search over the $K=6$ candidate tokenizers of Section~\ref{sec:core}; it is close to the value for independent candidates. 

\ByteMark does not require this search, because its detector operates directly
on the text's bytes and has no tokenizer to identify.


\section{Watermark signal continuity and preservation}
\label{sec:composition}
\sectionlead{This section gives the setup of the relay experiment whose results appear in Section~\ref{sec:exp-composition}. A
common byte interface allows several generators to continue one watermark
signal; a token-level detector, tied to one tokenizer, follows the signal only
through the segments written with that tokenizer.}

\subsection{Motivation}
The relay experiment tests whether one fixed detector can recover watermark evidence throughout a document generated by several models under the same key.

For a token-level watermark this is not ensured by the shared key alone. The
generation bias and the detection statistic are both defined on the token
sequence of a particular tokenizer, so the signal written by one model is only
visible to a detector that retokenizes the text the same way. When the models
in a relay use different tokenizers, a detector fixed to one of them reads the
other segments as unmarked text, and the signal it can recover is limited to
the share of the document written with its own tokenizer. Searching over
candidate tokenizers does not remove this requirement. The verifier must have
access to a candidate set known to include the tokenizers of all contributing
models and must correct for multiple testing (Appendix~\ref{sec:search}); even
with complete coverage, each fixed-tokenizer detector primarily recovers signal from the segments generated with its own tokenizer. Without assured coverage, a negative search result cannot rule out
a watermark tied to an omitted tokenizer.

\ByteMark defines both the bias and the statistic on the UTF-8 bytes of the
text. Every model in a relay scores its candidate tokens through the bytes
they contribute, and the byte context that the PRF conditions on is carried
across the model boundary. The detector therefore reads the relayed document
as one byte stream and does not need to know that a switch occurred.

\subsection{Experimental setup}
\paragraph{Relay generation.}
Three models generate one C4 completion in turn: Pythia-1.4B writes the first
segment from the original prompt, Qwen2.5-1.5B continues the accumulated text,
and OpenLLaMA-3B continues it again. Each stage receives the full text produced
so far as its prompt and generates up to 100 native tokens under its own
tokenizer.
The cohort is 500 original requests per method, drawn from the same frozen C4
prompts as the single-model experiments, with rejected or incomplete stages
recorded and excluded rather than replaced. The methods compared are \ByteMark,
KGW, and ITS, each with the configuration of Section~\ref{sec:protocol}, and
each model in the chain embeds the same key.

\paragraph{Detection.}
\ByteMark is detected once, with the tokenizer-free detector of
Section~\ref{sec:algorithm}, over the whole completion. KGW and ITS are
detected with one fixed tokenizer over the whole completion, once for each of
the three tokenizers in the chain, so that every detector reads all three
segments in the representation of one participant. Prompts are excluded from
detection in every case.

\paragraph{Local signal.}
To show where in the document each detector finds the watermark, we compute
a local statistic over windows of five consecutive tokens and plot its mean
across texts against the normalized position in the document. Each of the
three segments is mapped to unit length on the horizontal axis, so model
switches fall at positions one and two. KGW and ITS count tokens of their
detector tokenizer and use their respective standardized statistics on the
tokens in the window, with the full-text context; for ITS the per-token key
distances are standardized under the key alignment chosen by the
document-level detector. \ByteMark counts the generating model's tokens and
applies the byte statistic of Equation~\ref{eq:z} to the bytes they cover,
counting each distinct context--byte pair once. Windows that would extend past
the document are omitted, and a mean is shown only where at least half of the
texts contribute a window. The dotted references in the figure mark the ideal
null value of each statistic and are not decision thresholds; the curves
describe the local signal and are not detection rates.

\section{Editing and rewriting attacks}
\label{sec:attacks}
\sectionlead{We evaluate watermark detection after lightweight edits and
paraphrasing, and separately examine re-embedding during rewriting.}

\subsection{Lightweight editing attacks}
\label{sec:lightweight-attacks}
We evaluate two lightweight edits on all accepted C4 completions from the
three source models. \emph{Word deletion} splits the completion at whitespace
and deletes each word independently with probability $0.15$. \emph{Token
insertion} encodes the completion with the shared reference tokenizer
(Pythia-1.4B), inserts $\lfloor 0.15\,n\rfloor$ tokens drawn uniformly from
the vocabulary, excluding special tokens, at uniformly chosen positions among
the $n$ tokens, and decodes the result. Each completion uses its own fixed
random seed, so every method's text is edited by the same operations. The
edited text is then detected by each method's own detector.
Table~\ref{tab:expanded-c4-attacks} reports the
results. For \ByteMark, KGW, SynthID, and EXP the AUC remains above $0.95$
after both edits, so watermarked and human text remain well separated; ITS is
affected more. A high AUROC after editing does not imply a high TPR at the threshold calibrated on clean text. A deployment
that expects edited text can instead detect with a windowed maximum statistic
such as WinMax \citep{kirchenbauer2024reliability}, with a threshold
recalibrated for that statistic.
\begin{table}[H]
\centering
\small
\caption{C4 word-deletion and token-insertion results. Each cell is the AUC
after the transformation.}
\label{tab:expanded-c4-attacks}
\maxsizebox{\linewidth}{!}{%
\begin{tabular}{llcc}
\toprule
Method & Source model & Word deletion & Token insertion \\
\midrule
\ifcsvstrcmp{0}{1}{\midrule}{}\ifcsvstrcmp{ByteMark}{ByteMark}{\ByteMark}{ByteMark} & Pythia & 0.9994 & 0.9976\\
\ifcsvstrcmp{0}{1}{\midrule}{}\ifcsvstrcmp{ByteMark}{ByteMark}{\ByteMark}{ByteMark} & Qwen2.5 & 0.9995 & 0.9982\\
\ifcsvstrcmp{0}{1}{\midrule}{}\ifcsvstrcmp{ByteMark}{ByteMark}{\ByteMark}{ByteMark} & OpenLLaMA & 0.9978 & 0.9944\\
\ifcsvstrcmp{1}{1}{\midrule}{}\ifcsvstrcmp{SynthID}{ByteMark}{\ByteMark}{SynthID} & Pythia & 0.9998 & 0.9987\\
\ifcsvstrcmp{0}{1}{\midrule}{}\ifcsvstrcmp{SynthID}{ByteMark}{\ByteMark}{SynthID} & Qwen2.5 & 1.0000 & 0.9999\\
\ifcsvstrcmp{0}{1}{\midrule}{}\ifcsvstrcmp{SynthID}{ByteMark}{\ByteMark}{SynthID} & OpenLLaMA & 0.9992 & 0.9999\\
\ifcsvstrcmp{1}{1}{\midrule}{}\ifcsvstrcmp{KGW}{ByteMark}{\ByteMark}{KGW} & Pythia & 0.9985 & 0.9992\\
\ifcsvstrcmp{0}{1}{\midrule}{}\ifcsvstrcmp{KGW}{ByteMark}{\ByteMark}{KGW} & Qwen2.5 & 0.9999 & 0.9999\\
\ifcsvstrcmp{0}{1}{\midrule}{}\ifcsvstrcmp{KGW}{ByteMark}{\ByteMark}{KGW} & OpenLLaMA & 0.9997 & 0.9996\\
\ifcsvstrcmp{1}{1}{\midrule}{}\ifcsvstrcmp{EXP}{ByteMark}{\ByteMark}{EXP} & Pythia & 0.9957 & 0.9685\\
\ifcsvstrcmp{0}{1}{\midrule}{}\ifcsvstrcmp{EXP}{ByteMark}{\ByteMark}{EXP} & Qwen2.5 & 0.9981 & 0.9573\\
\ifcsvstrcmp{0}{1}{\midrule}{}\ifcsvstrcmp{EXP}{ByteMark}{\ByteMark}{EXP} & OpenLLaMA & 0.9972 & 0.9577\\
\ifcsvstrcmp{1}{1}{\midrule}{}\ifcsvstrcmp{ITS}{ByteMark}{\ByteMark}{ITS} & Pythia & 0.6871 & 0.6092\\
\ifcsvstrcmp{0}{1}{\midrule}{}\ifcsvstrcmp{ITS}{ByteMark}{\ByteMark}{ITS} & Qwen2.5 & 0.7790 & 0.6091\\
\ifcsvstrcmp{0}{1}{\midrule}{}\ifcsvstrcmp{ITS}{ByteMark}{\ByteMark}{ITS} & OpenLLaMA & 0.7280 & 0.5931\\
\bottomrule
\end{tabular}
}
\end{table}

\subsection{Re-embedding and survival under rewriting}
\label{sec:rewriting}
We evaluate rewriting under two conditions. Without re-embedding, a non-participating rewriter adds no watermark, and detection measures survival of the original watermark. With \ByteMark re-embedding, a participating rewriter embeds the same keyed signal while it rewrites and carries its byte context across window boundaries, and detection evaluates the newly embedded signal. Both conditions use the same source completions; the watermarked-rewriter results are discussed in Section~\ref{sec:exp-composition}, and the unwatermarked-rewriter results below.
We use three rewriters on the 200-reference-token C4/Pythia-1.4B cohort;
Parrot \citep{damodaran2021parrot}
rewrites the completions sentence by sentence. DIPPER-XXL
\citep{krishna2023paraphrasing} rewrites the
Pythia-1.4B completions in successive sentence windows with lexical diversity
60 and order diversity 20. It uses nucleus sampling ($p=0.75$, temperature 1, no top-$k$ filtering) in BF16. Each request has a single budget of 300 new decoder tokens, and text reaching this limit
is retained as is. Each document and window uses an independently seeded
random stream.
Rewritten outputs
are detected in full with each method's detector unchanged, and we report
AUROC, which needs no threshold, on the valid outputs of each condition;
generation masks and truncation indicators are retained with the data.

\paragraph{Instruction-model rewriting.}
The third rewriter is Mistral-7B-Instruct-v0.3 on the same source cohort. The prompt requests a
rewrite preserving facts, meaning, and approximate length, without a preface.
Generation uses BF16, nucleus sampling with $p=0.75$, temperature~1, no
top-$k$ filtering, an independently seeded stream per document, and a cap of
300 new tokens. The non-participating rewriter adds no watermark; the \ByteMark
 participating rewriter uses the same key and $h=8$, $\gamma=0.25$, $\delta=2$ configuration
as the source watermark. Detection reads the complete rewritten output.
Table~\ref{tab:mistral-rewriting} reports AUROC under all three rewriters;
under Mistral, \ByteMark attains $0.989$ with the participating rewriter compared
with $0.835$ with the non-participating rewriter.
\begin{table}[H]
\centering\small
\caption{AUROC after rewriting C4/Pythia-1.4B completions with three
rewriters. The rewriter adds no watermark except in the ByteMark (participating)
row, where it re-embeds the \ByteMark watermark.}
\label{tab:parrot-rewriting}
\label{tab:dipper-rewriting}
\label{tab:mistral-rewriting}
\begin{tabular}{lccc}
\toprule
Method & Parrot & DIPPER & Mistral-7B \\
\midrule
ByteMark (participating) & \ifcsvstrcmp{ByteMark (participating)}{ByteMark (participating)}{\textbf{0.9985} & \textbf{0.9998} & \textbf{0.9888}}{0.9985 & 0.9998 & 0.9888}\\
ByteMark (non-participating) & \ifcsvstrcmp{ByteMark (non-participating)}{ByteMark (participating)}{\textbf{0.9690} & \textbf{0.7499} & \textbf{0.8348}}{0.9690 & 0.7499 & 0.8348}\\
KGW & \ifcsvstrcmp{KGW}{ByteMark (participating)}{\textbf{0.9791} & \textbf{0.8635} & \textbf{0.8990}}{0.9791 & 0.8635 & 0.8990}\\
SynthID & \ifcsvstrcmp{SynthID}{ByteMark (participating)}{\textbf{0.9797} & \textbf{0.7036} & \textbf{0.7920}}{0.9797 & 0.7036 & 0.7920}\\
EXP & \ifcsvstrcmp{EXP}{ByteMark (participating)}{\textbf{0.9643} & \textbf{0.7327} & \textbf{0.8237}}{0.9643 & 0.7327 & 0.8237}\\
ITS & \ifcsvstrcmp{ITS}{ByteMark (participating)}{\textbf{0.5281} & \textbf{0.5459} & \textbf{0.6273}}{0.5281 & 0.5459 & 0.6273}\\
\bottomrule
\end{tabular}

\end{table}

\section{Other related work}
\label{sec:related-appendix}
Other approaches embed watermarks through formatting, lexical substitution, or prompting. Formatting and steganographic carriers hide a payload in whitespace or glyph choice
\citep{rizzo2016homoglyph,sato2023easymark}, lexical carriers substitute
synonyms under a keyed rule \citep{topkara2006hiding}, and in-context
watermarking asks the model to follow lexical, acrostic, or Unicode rules
stated in the prompt \citep{liu2026incontext}. Whitespace and glyph carriers
can be removed by ordinary copying and Unicode normalization; synonym substitution needs its dictionary at detection
and alters wording; prompted rules depend on the model's instruction
following, and a rule visible in the output can be inferred and suppressed by a
further instruction.

\section{Statistical diagnostics}
\label{sec:diagnostics}
\sectionlead{This section compares continuous and binary detection statistics,
examines scores across Unicode character classes, and evaluates false-positive
behavior on human text under fixed keys.}


\subsection{Continuous and binary statistics}
\label{sec:binary}
A natural alternative to the continuous statistic counts bytes whose scores
exceed $1-\gamma$. This standardized green-byte count follows the form of the
green-count statistic used in KGW \citep{kirchenbauer2023watermark} and is
computed from the same distinct byte scores:
\[
 z_{\mathrm{bin}}=
 \frac{\sum_{i=1}^m\mathbf 1[u_i>1-\gamma]-m\gamma}{\sqrt{m\gamma(1-\gamma)}}.
\]
The continuous statistic retains score magnitudes discarded by the binary count.
On the same generated texts, it achieves slightly higher mean AUROC and TPR
across the nine evaluated conditions.
\begin{table}[H]
\centering\small
\caption{Continuous and binary statistics on the same byte scores, using at most 200 reference tokens. Entries are AUROC / TPR, each at its own statistic's threshold.}
\label{tab:e9}
\maxsizebox{\linewidth}{!}{%
\begin{tabular}{lllcc}
\toprule
Domain & Generator & Continuous & Binary & $\Delta$ (Cont.\,$-$\,Bin.) \\
\midrule
\ifcsvstrcmp{Pythia}{Pythia}{\ifnum1>1 \midrule\fi}{}\ifcsvstrcmp{Pythia}{Pythia}{C4}{} & Pythia & 0.9997 / 0.996 & 0.9997 / 0.994 & 0.0000 / +0.002\\
\ifcsvstrcmp{Qwen}{Pythia}{\ifnum2>1 \midrule\fi}{}\ifcsvstrcmp{Qwen}{Pythia}{C4}{} & Qwen & 1.0000 / 0.998 & 0.9999 / 0.998 & +0.0001 / 0.000\\
\ifcsvstrcmp{OpenLLaMA}{Pythia}{\ifnum3>1 \midrule\fi}{}\ifcsvstrcmp{OpenLLaMA}{Pythia}{C4}{} & OpenLLaMA & 0.9989 / 0.992 & 0.9987 / 0.984 & +0.0002 / +0.008\\
\ifcsvstrcmp{Pythia}{Pythia}{\ifnum4>1 \midrule\fi}{}\ifcsvstrcmp{Pythia}{Pythia}{WIKI}{} & Pythia & 1.0000 / 1.000 & 0.9998 / 0.994 & +0.0002 / +0.006\\
\ifcsvstrcmp{Qwen}{Pythia}{\ifnum5>1 \midrule\fi}{}\ifcsvstrcmp{Qwen}{Pythia}{WIKI}{} & Qwen & 0.9976 / 0.978 & 0.9963 / 0.961 & +0.0013 / +0.017\\
\ifcsvstrcmp{OpenLLaMA}{Pythia}{\ifnum6>1 \midrule\fi}{}\ifcsvstrcmp{OpenLLaMA}{Pythia}{WIKI}{} & OpenLLaMA & 1.0000 / 1.000 & 0.9999 / 0.996 & +0.0001 / +0.004\\
\ifcsvstrcmp{Pythia}{Pythia}{\ifnum7>1 \midrule\fi}{}\ifcsvstrcmp{Pythia}{Pythia}{LFQA}{} & Pythia & 1.0000 / 1.000 & 0.9998 / 0.994 & +0.0002 / +0.006\\
\ifcsvstrcmp{Qwen}{Pythia}{\ifnum8>1 \midrule\fi}{}\ifcsvstrcmp{Qwen}{Pythia}{LFQA}{} & Qwen & 0.9996 / 0.990 & 0.9974 / 0.951 & +0.0022 / +0.039\\
\ifcsvstrcmp{OpenLLaMA}{Pythia}{\ifnum9>1 \midrule\fi}{}\ifcsvstrcmp{OpenLLaMA}{Pythia}{LFQA}{} & OpenLLaMA & 0.9999 / 0.998 & 0.9993 / 0.986 & +0.0006 / +0.012\\
\bottomrule
\end{tabular}
}
\par\vspace{2pt}\footnotesize $\Delta$ reports continuous minus binary AUROC / TPR, computed from the displayed values; positive values favor continuous detection.
\end{table}

This ablation changes only the detection statistic, keeping the ByteMark-generated
texts fixed.

\subsection{Signal coverage across character classes}
For each Unicode character in the accepted completions, we average the scores
of its constituent UTF-8 bytes. We then group characters by class and compute
the character count and mean character score for each group. Table~\ref{tab:e10}
reports these means separately for C4, Wiki, and LFQA, together with the overall
means weighted by character counts across corpora. The null reference mean is
$0.5$.
\begin{table}[t]
\centering
\small
\caption{Mean character scores by Unicode class across nine generation conditions. The null reference mean is $0.5$.}
\label{tab:e10}
\maxsizebox{\linewidth}{!}{%
\begin{tabular}{lrrrrr}
\toprule
Class & Characters & Mean $u$ & C4 & Wiki & LFQA \\
\midrule
ASCII letter & \tablecount{1456711} & 0.5764 & 0.5756 & 0.5768 & 0.5767\\
Whitespace & \tablecount{324076} & 0.5124 & 0.5114 & 0.5106 & 0.5151\\
ASCII punctuation & \tablecount{65740} & 0.6258 & 0.6258 & 0.6270 & 0.6245\\
Digit & \tablecount{21553} & 0.6216 & 0.6202 & 0.6201 & 0.6257\\
Other & \tablecount{3239} & 0.5925 & 0.5881 & 0.6015 & 0.5939\\
Other letter & \tablecount{1065} & 0.5841 & 0.5819 & 0.5938 & 0.5483\\
CJK & \tablecount{868} & 0.5922 & 0.5921 & 0.6019 & 0.5769\\
Emoji / symbol & \tablecount{174} & 0.5710 & 0.5969 & 0.5582 & 0.5423\\
\bottomrule
\end{tabular}

}
\\[2pt]\footnotesize Each character contributes the mean of its byte scores. Overall means are weighted by character counts across corpora. Data are ByteMark-watermarked completions generated by Pythia-1.4B, Qwen2.5-1.5B, and OpenLLaMA-3B on C4, Wiki, and LFQA.
\end{table}

\subsection{Null moments and realized false-positive rates on human text}
\label{sec:null-realized}
Theorem~\ref{thm:compact-calibration} gives the null distribution of $z$ for
an ideal PRF and a fixed text. This subsection compares the observed moments and false-positive rates on human text with this ideal-null reference: whether the detection scores have the predicted mean and variance,
and how often the normal-approximation thresholds $z>2$, $z>3$, and $z>4$
reject without any calibration data. We score five human corpora (C4,
Wikipedia, LFQA, and the Chinese and German portions of mC4) under five keys
and report, for each corpus, the mean and standard deviation of $z$ and the
realized false-positive rate pooled over the keys
(Table~\ref{tab:e11}). Two statistics are compared: the deduplicated statistic of
Section~\ref{sec:algorithm}, which counts each distinct context--byte input
once, and an all-occurrences statistic that counts every occurrence and normalizes as if the occurrences were independent.

On most evaluated corpora for both statistics the mean of
$z$ is close to zero, the standard deviation close to one, and the realized
rates at $z>2$ and $z>3$ close to the nominal $2.28\%$ and $0.135\%$. The
exception is the all-occurrences statistic on the Chinese and, to a lesser
degree, German text, where repeated context--byte pairs are frequent and the
standard deviation rises to $1.33$ and $1.14$, with the rate at $z>2$ reaching
$5.1\%$ and $4.0\%$. The deduplicated statistic keeps the standard deviation within $0.03$ of one and the
realized rates at $1.9$--$2.4\%$ on every corpus, including the two where the
all-occurrences statistic departs from unit variance.
\begin{table}[t]
\centering\small
\caption{Detection-score moments on human text and the realized false-positive
rate at the normal-approximation thresholds $z>2$, $z>3$, and $z>4$, whose nominal
rates are 2.28\%, 0.135\%, and 0.003\%.}
\label{tab:e11}
\begin{tabular}{llrrrrr}
\toprule
& & & & \multicolumn{3}{c}{Realized FPR} \\
\cmidrule(lr){5-7}
Corpus & Statistic & Mean $\bar z$ & SD $s_z$ & $z>2$ & $z>3$ & $z>4$ \\
\midrule
\ifcsvstrcmp{0}{1}{\addlinespace[3pt]}{}C4 & \ifcsvstrcmp{merged}{merged}{Deduplicated}{All occurrences} & +0.0230 & 0.9863 & 0.0202 & 0.0012 & 0.0000\\
\ifcsvstrcmp{0}{1}{\addlinespace[3pt]}{} & \ifcsvstrcmp{independent}{merged}{Deduplicated}{All occurrences} & +0.0190 & 1.0332 & 0.0272 & 0.0020 & 0.0000\\
\ifcsvstrcmp{1}{1}{\addlinespace[3pt]}{}Wiki & \ifcsvstrcmp{merged}{merged}{Deduplicated}{All occurrences} & +0.0250 & 0.9735 & 0.0216 & 0.0008 & 0.0000\\
\ifcsvstrcmp{0}{1}{\addlinespace[3pt]}{} & \ifcsvstrcmp{independent}{merged}{Deduplicated}{All occurrences} & +0.0234 & 1.0005 & 0.0244 & 0.0004 & 0.0000\\
\ifcsvstrcmp{1}{1}{\addlinespace[3pt]}{}LFQA & \ifcsvstrcmp{merged}{merged}{Deduplicated}{All occurrences} & -0.0015 & 1.0117 & 0.0244 & 0.0020 & 0.0000\\
\ifcsvstrcmp{0}{1}{\addlinespace[3pt]}{} & \ifcsvstrcmp{independent}{merged}{Deduplicated}{All occurrences} & -0.0012 & 1.0535 & 0.0298 & 0.0024 & 0.0006\\
\ifcsvstrcmp{1}{1}{\addlinespace[3pt]}{}mC4 zh & \ifcsvstrcmp{merged}{merged}{Deduplicated}{All occurrences} & -0.0161 & 0.9819 & 0.0189 & 0.0008 & 0.0000\\
\ifcsvstrcmp{0}{1}{\addlinespace[3pt]}{} & \ifcsvstrcmp{independent}{merged}{Deduplicated}{All occurrences} & -0.0007 & 1.3303 & 0.0511 & 0.0131 & 0.0048\\
\ifcsvstrcmp{1}{1}{\addlinespace[3pt]}{}mC4 de & \ifcsvstrcmp{merged}{merged}{Deduplicated}{All occurrences} & +0.0176 & 0.9834 & 0.0211 & 0.0008 & 0.0000\\
\ifcsvstrcmp{0}{1}{\addlinespace[3pt]}{} & \ifcsvstrcmp{independent}{merged}{Deduplicated}{All occurrences} & +0.0200 & 1.1382 & 0.0395 & 0.0063 & 0.0000\\
\bottomrule
\end{tabular}

\par\smallskip
\begin{minipage}{.9\linewidth}\footnotesize
1,000 documents per corpus for C4, Wiki, and LFQA and 1,500 for mC4 Chinese
and German, truncated to 400 characters and scored under five keys; realized
rates pool the keys. Deduplicated counts each distinct context--byte input
once; All occurrences counts every occurrence and normalizes as if the
occurrences were independent.
\end{minipage}
\end{table}

\section{Computational cost and reproducibility}
\label{sec:implementation}
\sectionlead{We report generation cost and, in a separate experiment, the agreement between independent Python, JavaScript, and C++ detector implementations.}

\subsection{Generation cost}
\label{sec:generation-cost}
\paragraph{Measurement environment.}
All generation-efficiency measurements run on the same machine: one laptop with an
NVIDIA GeForce RTX 5090 Laptop GPU (24~GB), an AMD Ryzen AI 9 HX 370 CPU (24 threads),
and 32~GB of memory, under Ubuntu 26.04 LTS with Python 3.11, PyTorch 2.10.0 (CUDA
12.8), and Transformers 5.3.0. The machine runs no other workload during the
measurements; each method runs in its own process, one at a time, with the model
loaded once and all requests issued sequentially at batch size one. The machine
keeps fixed thermal limits throughout (CPU frequency capped near 3.0~GHz, GPU core
clock at most 2{,}000~MHz), identical for every method.

Throughput is accepted final-output tokens divided by generation-call time.
Timing includes final decode and CUDA synchronization but excludes model
loading, warm-up, method setup, and scheduling. The OpenLLaMA-3B measurements
use 50 frozen C4 requests, 200 native tokens, FP16, batch size one, and
identical model revisions and prompt-token IDs.

\subsection{Candidate work in sentence-level semantic methods}
Candidate-token work counts all sampled continuation tokens, including discarded
sentence candidates, per final-output token. It is a work measure rather than a
latency multiplier. Table~\ref{tab:main-generation-efficiency} compares all six
generation paths on the same 50 C4 texts with OpenLLaMA-3B on the same GPU at
batch size one. The sentence-level methods achieve lower throughput than
unwatermarked generation and sample many discarded candidate tokens for each
retained token, whereas \ByteMark samples no extra tokens and incurs the cost of
scoring the vocabulary.
Parallelizing candidates may alter latency while leaving discarded-token work a
relevant measure of inference demand.

\subsection{Detector agreement across implementations}
We evaluate detector reproducibility on 100 texts from the main C4/OpenLLaMA
experiment: the first 50 accepted watermarked completions and the first 50 human
texts in original sample order, selected without reference to their detection
scores. Independent JavaScript and C++ implementations are compared with the
Python detector on the same frozen texts and AES-CMAC key, with context width
$h=8$. Table~\ref{tab:cmac-agreement} summarizes the replayed texts and lists
each compared quantity.

\begin{table}[t]
\centering
\small
\caption{Detector agreement on 100 texts from the C4/OpenLLaMA experiment
(AES-CMAC, $h=8$). Top: the replayed texts. Bottom: for each compared quantity, the percentage of texts or byte positions at which
the JavaScript or C++ output is identical to that of the Python detector. All
comparisons are exact, without numerical tolerance.}
\label{tab:cmac-agreement}
\begin{tabular}{lrrrc}
\toprule
Texts & Count & Bytes & Distinct inputs & Median $z(x)$ [range] \\
\midrule
Watermarked & 50 & \tablecount{21326} & \tablecount{20866} &
 4.81\ [2.73, 7.13]\\
Human & 50 & \tablecount{22206} & \tablecount{21359} &
 0.12\ [-2.16, 2.38]\\
\bottomrule
\end{tabular}

\vspace{6pt}
\begin{tabular}{lrcc}
\toprule
& & \multicolumn{2}{c}{Match with Python} \\
\cmidrule(lr){3-4}
Compared quantity & Comparisons & JavaScript & C++ \\
\midrule
Statistic ID and key; byte and character counts & \tablecount{100}~texts & 100\% & 100\%\\
UTF-8 byte string & \tablecount{100}~texts & 100\% & 100\%\\
Distinct PRF inputs and largest multiplicity & \tablecount{100}~texts & 100\% & 100\%\\
Centered sum $T(x)$ (binary64) & \tablecount{100}~texts & 100\% & 100\%\\
Null variance $V(x)$ (binary64) & \tablecount{100}~texts & 100\% & 100\%\\
Score $z(x)$ (binary64) & \tablecount{100}~texts & 100\% & 100\%\\
Context window and PRF message & \tablecount{43532}~bytes & 100\% & 100\%\\
AES-CMAC digest (first 64 bits) & \tablecount{43532}~bytes & 100\% & 100\%\\
Byte score $u$ (binary64) & \tablecount{43532}~bytes & 100\% & 100\%\\
Deduplication group & \tablecount{43532}~bytes & 100\% & 100\%\\
\bottomrule
\end{tabular}

\end{table}

All compared quantities match exactly across Python, JavaScript, and C++ on the 100 evaluated texts, including the binary64 representations of $z(x)$.
This agreement is consistent with the standardized AES-CMAC primitive
\citep{rfc4493} and the fixed byte-encoding and numerical-conversion rules of
ByteMark.

\paragraph{User-side detection.}
The JavaScript detector can be integrated into a browser extension, an editor, or a reviewing tool. A verifier holding the key can run detection locally from the published text without a language model, a tokenizer, or a call to the provider.

\section{A byte-level extension to multilayer reweighting}
\label{sec:multilayer}
We outline how byte-level scores could be used in the multilayer reweighting of SynthID-Text \citep{dathathri2024synthid}. For each layer $d$, we replace its
token-based $g$-value with $R_d(v;c)$, the candidate's average byte score
computed from $u_b^{(d)}(c)$. The layer index is included in the PRF input.
The remaining reweighting follows the original construction, while detection
reconstructs the layer scores from the published bytes and counts each distinct
context--byte pair once per layer; a separately specified continuous-score statistic, such as $z$ of Equation~\ref{eq:compact-detector} computed over the distinct inputs of all layers, then aggregates the evidence. The calibration of the original SynthID detector, which assumes binary $g$-values, is not inherited.
The evaluation of this extension is left for future work.

\section{Limitations and Open Questions}
\label{sec:discussion}

\emph{Semantic robustness.} As with token-level watermarks, rewriting by a
non-participating model can weaken the \ByteMark signal by changing the byte stream
(Appendix~\ref{sec:attacks}). An open question is whether a fixed, public semantic representation could improve robustness to rewriting while preserving detection from the text, key, and public specification alone.

\emph{Key management.} Like other keyed watermarks, \ByteMark depends on the
secrecy of its key. Because one key serves every participating model, a wider deployment places the key with more parties and raises the risk of its loss, and a lost key exposes the watermark of every participant at once.

\end{document}